\documentclass[cmp,numbook,envcountsame]{svjour}  
\usepackage{amsmath,amssymb,mathtools,bm,lipsum}
\usepackage{amsfonts,amssymb}

\journalname{Communications in Mathematical Physics}

\newlength{\Taille}

\newcommand{\flechebas}[1]{
  \settoheight{\unitlength}{\mbox{$#1$}}
  \settowidth{\Taille}{\mbox{~${\scriptstyle #1}$}}
  \addtolength{\unitlength}{4ex}
  \begin{picture}(0,1)
    \put(0,1){\vector(0,-1){1}}
    \put(0,0.5){\makebox(0,0){${\scriptstyle #1}$ \hspace{\the\Taille}}}
  \end{picture}}
\newcommand{\flechehaut}[1]{
  \settoheight{\unitlength}{\mbox{$#1$}}
  \settowidth{\Taille}{\mbox{~${\scriptstyle #1}$}}
  \addtolength{\unitlength}{4ex}
  \begin{picture}(0,1)
    \put(0,0){\vector(0,1){1}}
    \put(0,0.5){\makebox(0,0){\hspace{\the\Taille}${\scriptstyle #1}$ }}
  \end{picture}}
\newcommand{\flechedroite}[1]{
  \settowidth{\unitlength}{\mbox{$#1$}}
  \settoheight{\Taille}{\mbox{${\scriptstyle #1}$}}
  \addtolength{\Taille}{1ex}
  \addtolength{\unitlength}{4ex}
  \raisebox{0.5ex}{
  \begin{picture}(1,0)
    \put(0,0){\vector(1,0){1}}
    \put(0.5,0){\makebox(0,0){${\scriptstyle #1}$ \vspace{\the\Taille}}}
  \end{picture}}}
\newcommand{\flechegauche}[1]{
  \settowidth{\unitlength}{\mbox{$#1$}}
  \settoheight{\Taille}{\mbox{${\scriptstyle #1}$}}
  \addtolength{\Taille}{1ex}
  \addtolength{\unitlength}{4ex}
  \raisebox{0.5ex}{
  \begin{picture}(1,0)
    \put(1,0){\vector(-1,0){1}}
    \put(0.5,0){\makebox(0,0){${\scriptstyle #1}$ \vspace{\the\Taille}}}
  \end{picture}}}

\newcommand{\cqfd}{\hfill $\square$}

\newcommand{\RR}{\mathbb{R}}

\newcommand{\NN}{\mathbb{N}}
\newcommand{\CC}{\mathbb{C}}

\newcommand{\cH}{\mathcal{H}}
\newcommand{\cF}{\mathcal{F}}
\newcommand{\cD}{\mathcal{D}}
\newcommand{\cK}{\mathcal{K}}
\newcommand{\cB}{\mathcal{B}}
\newcommand{\cE}{\mathcal{E}}
\newcommand{\cJ}{\mathcal{J}}

\newcommand{\cC}{\mathcal{C}}

\newcommand{\cA}{\mathcal{A}}
\newcommand{\cS}{\mathcal{S}}
\newcommand{\cN}{\mathcal{N}}
\newcommand{\cI}{\mathcal{I}}

\newcommand{\cM}{\mathcal{M}}

\begin{document}

\title{{Non-existence of ground states in the translation invariant Nelson model}}
\titlerunning{Non-existence of ground states in the translation invariant Nelson model}

\author{Thomas Norman Dam}
\institute{Aarhus Universitet, Nordre Ringgade 1, 8000 Aarhus C  Denmark.\\ \email{cyperman@gmail.com}}
\authorrunning{Thomas Norman Dam}

\date{}

\maketitle

\begin{abstract}
In this paper, the translation invariant massless Nelson model with ultraviolet cutoff is investigated. It is proven, that the fiber operators have no ground state if the interaction is not infrared regular. Similar results have been obtained by other authors under severe assumptions on the regularity of the mass shell. These assumptions are removed in this paper and we use only rotation invariance and non-degeneracy of ground states to obtain our result.
\end{abstract}

\section{Introduction}
In this paper, we study the translation invariant massless Nelson model with ultraviolet cutoff. The model can (after a unitary transformation) be written as a direct integral of fiber operators $\{ H_\mu(\xi) \}_{\xi\in \RR^3}$. The spectral properties of these operators were first investigated by by J. Fr\"ohlich in his Phd-thesis, which was published in the two papers \cite{Frlich1} and \cite{Frlich2}. Fr\"ohlich showed, that if the field is massive or there is an infrared cut-off then $H_\mu(\xi)$ has a ground state for $\xi$ in an open ball around $0$. He also proved, that if the field is massless and a ground state exists for sufficiently many $H_\mu(\xi)$, then one can reach some physically unacceptable conclusions. The aim of this paper is to prove that $H_\mu(\xi)$ does not have a ground state if the field is massless and no infrared conditions are assumed. 

In the paper \cite{Pizzo}, it is proven, that ground states exists in a non-equivalent Fock representation. A consequence of this result is that the usual "taking the massgap to 0" strategy for proving existence of ground states does not work. This strongly indicates that there should be no ground state.

A proof of absence of ground states in a similar model was given by I. Herbst and D. Hasler in the paper \cite{HerbstHasler}. They consider the fiber operators of the massless and translation invariant Pauli-Fierz model $\{ \cH(\xi) \}_{\xi\in \RR^3}$. They prove that $\cH(\xi_0)$ has no ground state if $\xi\mapsto \inf(\sigma(\cH(\xi)))$ is differentiable at $\xi_0$ and has a non-zero derivative. One may easily work out the same problem for the Nelson model and obtain the same conclusions. However, proving the existence of a non-zero derivative is an extremely hard problem and such a result has only been achieved for weak coupling and small $\xi$ (see \cite{Hasler}). Furthermore, $\xi=0$ is a global minimum for $\xi\mapsto \inf(\sigma(H_\mu(\xi)))$ and therefore the derivative (if it exists) must be 0. However, $H_\mu(0)$ has no ground states as we shall prove below.

In fact, we shall prove that $H_\mu(\xi)$ has no ground state for any non-zero coupling strength $\mu$ and $\xi\in \RR^3$. Our proof is based on strategy used by I. Herbst and D. Hasler, but we remove the assumption regarding the existence of a non-zero derivative. Instead, we use rotation invariance of the map $\xi\mapsto \inf(\sigma(H_\mu(\xi)))$, non degeneracy of ground states and the HVZ-theorem to reach our conclusion.

\section{Notation and preliminaries}
Let $(\cM,\cF,\mu)$ be a $\sigma$-finite measure space and $X$ be a separable Hilbert space. We will write $L^2(\cM,\cF,\mu,X)=L^2(\cM,\cF,\mu)\otimes X$ for the Hilbert space valued $L^2$-space where we define $f\otimes \psi=(k\mapsto f(k)\psi)$ for $f\in L^2(\cM,\cF,\mu)$ and $\psi\in X$. In case $\cM$ is a topological space we will write $\cB(\cM)$ for the Borel $\sigma$-algebra.

Fix $\nu\in \NN$. Let $\cH=L^2(\RR^\nu,\cB(\RR^\nu),\lambda_\nu)$ where $\lambda_\nu$ is the Lebesgue measure. For $A\in \cB(\RR^\nu)$ measurable we define $\cH_A=L^2(\RR^\nu,\cB(\RR^\nu),\lambda_\nu^A)$ where $\lambda_\nu^A=1_A\lambda_\nu$. For $n\in \NN_0$ we write $\cH_A^{\otimes n}$ for the $n$-fold tensor product. Note that
\begin{align*}
\cH_A^{\otimes n}&=L^2(\RR^{n\nu},\cB(\RR^{n\nu}),1_{A^n}\lambda_{n\nu})
\end{align*}
for $n\geq 1$ and $\cH_A^{\otimes 0}=\CC$ by definition. For $n\geq 1$ we write $\cS_n$ for the set of permutations of $\{1,\dots,n\}$ and define the symmetric projection 
\begin{align*}
(S_nf)(k_1,\dots,k_n)=\frac{1}{n!}\sum_{\sigma\in \cS_n} f(k_{\sigma(1)},\dots,k_{\sigma(n)}).
\end{align*}
We note that $f\in S_n (\cH_A^{\otimes n})=\cH_A^{\otimes_s n}$ if and only if $f\in \cH_A^{\otimes n}$ and $f(k_1,\dots,k_n)=f(k_{\sigma(1)},\dots,k_{\sigma(n)})$ for any $\sigma\in \cS_n$. To include $n=0$ we define $S_0=1$ on $S_0 \cH_A^{\otimes 0}=\CC=\cH_A^{\otimes_s 0}$. The bosonic Fock space is defined by
\begin{equation*}
\cF(\cH_A)=\bigoplus_{n=0}^\infty \cH_A^{\otimes_s n}.
\end{equation*}
We will write an element $\psi\in \cF(\cH_A)$ in terms of its coordinates as $\psi=(\psi^{(n)})$ and define the vacuum $\Omega=(1,0,0,\dots)$. Furthermore, for $\cD\subset \cH$ and $f_1,\dots,f_n\in \cH$ we introduce the notation
\begin{align*}
S_n(f_1\otimes\cdots\otimes f_n)&=f_1\otimes_s\cdots \otimes_s f_n,\\
\epsilon(f_i)&=\sum_{n=0}^{\infty}\frac{f_i^{\otimes n}}{\sqrt{n!}},\\
\cJ(\cD)&=\{ \Omega \}\cup \{f_1\otimes_s\cdots \otimes_s f_n\mid f_i\in \cD,n\in \mathbb{N} \} \,\,\,\, \text{and} \\  \cE(\cD)&= \{ \epsilon(f)\mid f\in \cD \}
\end{align*}
where $f_i^{\otimes 0}=\Omega$. One may prove that if $\cD\subset \cH$ is dense then $ \cE(\cD)$ is a total subset of $\cF(\cH)$. From this one easily concludes $ \cJ(\cD)$ is total. 

For $g\in \cH_A$ one defines the annihilation operator $a(g)$ and creation operator $a^{\dagger}(g)$ on symmetric tensors in $\cF(\cH_A)$ using $a(g)\Omega=0,a^\dagger(g)\Omega=g$ and
\begin{align*}
a(g)( f_1\otimes_s\cdots\otimes_s f_n )&=\frac{1}{\sqrt{n}}\sum_{i=1}^{n} \langle g,f_i \rangle f_1\otimes_s\cdots\otimes_s \widehat{f}_i\otimes_s\cdots\otimes_s f_n\\
a^\dagger(g)( f_1\otimes_s\cdots\otimes_s f_n )&=\sqrt{n+1}g\otimes_s f_1\otimes_s\cdots\otimes_s f_n
\end{align*}
where $\widehat{f}_i$ means that this element is omitted. One can show that these operators extends to closed operators on $\cF(\cH_A)$ and that $(a(g))^*=a^{\dagger}(g)$. Furthermore, we have the canonical commutation relations which states
\begin{equation*}
\overline{[a(f),a(g)]}=0=\overline{[a^\dagger(f),a^\dagger(g)]} \,\,\text{and}\,\,\, \overline{[a(f),a^\dagger(g)]}=\langle f,g\rangle.
\end{equation*}
One now introduces the selfadjoint field operators
\begin{equation*}
\varphi(g)=\overline{ a(g)+a^\dagger(g) }.
\end{equation*}
Let $\omega:\RR^\nu\rightarrow \RR$ be a map. We define $\omega^{(n)}(k_1,\dots,k_n)=\omega(k_1)+\cdots+\omega(k_n)$, $\omega^{(0)}=0$ and
\begin{equation}\label{Sumdecomp}
d\Gamma_A(\omega)=\bigoplus_{n=0}^{\infty}d\Gamma_A^{(n)}(\omega),
\end{equation}
where $d\Gamma_A^{(n)}(\omega)$ is multiplication by $\omega^{(n)}$ on $\cH_A^{\otimes n}$. The number operator is defined as $N_A=d\Gamma_A(1)$. Let $U$ be a bounded map from $\cH_A$ to $\cH_A$ with norm smaller than $1$. Then we define the map
\begin{equation*}
\Gamma(U)=1\oplus \bigoplus_{n=1}^\infty  U\otimes\cdots\otimes U\mid_{\cH_A^{\otimes_s n}}=1\oplus \bigoplus_{n=1}^\infty  \Gamma^{(n)}(U).
\end{equation*}
on $\cF(\cH_A)$. If $U$ is unitary then $\Gamma(U)$ is unitary as well. The following lemma can be found in \cite{Hirokawa1}:
\begin{lemma}\label{Lem:FundamentalIneq}
	Assume $\omega> 0$ almost everywhere. If $g\in \cH_A$ and $\omega^{-1/2}g \in \cH_A$ then $\varphi(g)$, $a^\dagger(g)$ and $a(g)$ are $d\Gamma_A(\omega)^{1/2}$ bounded and
	\begin{equation*}
	\lVert \varphi(g) \psi \lVert\leq 2 \lVert (\omega^{-1/2}+1)g \lVert  \lVert (d\Gamma_A(\omega)+1)^{1/2}\psi \lVert   
	\end{equation*}
	for all $\psi\in \cD(d\Gamma_A(\omega)^{1/2})$. In particular, $\varphi(g)$ is infinitesimally $d\Gamma(\omega)$ bounded. Furthermore, $d\Gamma(\omega)+\varphi(g)\geq  -\lVert \omega^{-1/2}g \lVert^2$. 
\end{lemma}
 We have the following obvious lemma which is useful for calculations
\begin{lemma}\label{Lem:FundamentalCalculations}
Let $f,g\in \cH$. Then $\epsilon(g)\in \cD(N^n)$ for all $n\geq 0$. Furthermore:
\begin{enumerate}
\item[\textup{(1)}]  $a(g)\epsilon(f)=\langle g,f \rangle \epsilon(f)$ and $\langle \varepsilon(g),\varepsilon(f) \rangle=e^{\langle g,f\rangle}$.

\item[\textup{(2)}] If $f\in \cD(\omega)$ then $\epsilon(f)\in \cD(d\Gamma(\omega))$ and $d\Gamma(\omega)\epsilon(f)=a^\dagger(\omega f)\epsilon(f)$. In particular, $\langle\epsilon(g),d\Gamma(\omega)\epsilon(f) \rangle=\langle g,\omega f\rangle e^{\langle g,f\rangle}$.
\end{enumerate}
\end{lemma}
Define
\begin{align*}
\cC\cS_A=\{ f\in  \cH_A\mid \text{$\exists R>0$ such that $1_{B_R(0)\cap A}f=1_A f$ almost everywhere} \}.
\end{align*}
which is a dense subspace inside $\cH_A$. We will also need the contraction $P_A:\cH_{\RR^\nu}\rightarrow \cH_A$ defined by
\begin{align*}
P_A(v)=v
\end{align*}
almost everywhere on $A$. Let $m_i: \RR^\nu\rightarrow \RR$ be the $i$'th coordinate projection and define for $n\in \NN_0$
\begin{align*}
d\Gamma_A(m)=(d\Gamma_A(m_1),\dots,d\Gamma_A(m_{\nu}))\\
d\Gamma_A^{(n)}(m)=(d\Gamma_A^{(n)}(m_1),\dots,d\Gamma_A^{(n)}(m_{\nu}))
\end{align*}
For $n\geq 1$ we define $m^{(n)}:(\RR^{\nu})^n\rightarrow \RR^{\nu}$ by $m^{(n)}(k)=k_1+\dots+k_n$. Then for $K:\RR^\nu\rightarrow \RR$ we have
\begin{align*}
K(\xi-d\Gamma_A(m))=K(\xi)\oplus \bigoplus_{n=0}^\infty K(\xi-d\Gamma_A^{(n)}(m))
\end{align*}
where $K(\xi-d\Gamma_A^{(n)}(m))$ is to be interpreted as the multiplication operator on $\cH_A^{\otimes_s n}$ defined by the map $k\mapsto K(\xi-m^{(n)}(k))$. If $A=\RR^\nu$ then we will leave $A$ out in all of the above notation.

Let $B=(B_1,\dots,B_\nu)$ be a tuple of operators on $\cH_A$ and define for any $k\in \RR^{\nu}$ an operator on $\cD(B):=\cap_{i=1}^\nu\cD(B_i)$ by
\begin{align*}
k\cdot B=\sum_{i=1}^{\nu}k_iB_i.
\end{align*}
For $\psi\in \cD(B)$ and $k\in \RR^\nu$ we have the following useful inequality 
\begin{align}\nonumber
\lVert k\cdot B\psi\lVert^2&=\sum_{i,j=1}^{\nu}\langle k_iB_i\psi,k_jB_j\psi  \rangle\leq \sum_{i,j=1}^{\nu}\lvert k_i\lvert \lvert k_j\lvert \lVert B_i\psi \lVert\lVert B_j\psi \lVert\\&\leq \sum_{i,j=1}^{\nu} \frac{1}{2}\lvert k_i\lvert^2 \lVert B_j\psi \lVert^2 +\frac{1}{2}\lvert k_j\lvert\lVert B_i\psi \lVert^2=\lvert k\lvert^2\lVert B\psi\lVert^2\label{eq:Vectorcauchy}
\end{align}

\section{The operator - basic properties and the main result}
Fix $K,\omega:\RR^\nu\rightarrow [0,\infty)$ measurable and let $v\in \cH$. Define for $A\in \cB(\RR^\nu)$ and $\xi\in \RR^\nu$ the Hamiltonian
\begin{align*}
H_\mu(\xi,A)=K(\xi-d\Gamma_A(m))+d\Gamma_A(\omega)+\mu \varphi(v_A)
\end{align*} 
where $v_A=P_A(v)$. We will drop $A$ from the notation if $A=\RR^{\nu}$.
\begin{lemma}\label{Lem:BasicProperties}
Assume $\omega>0$ almost everywhere, $v\in \cD(\omega^{-1/2})$ and $A\in \cB(\RR^\nu)$. Then $H_\mu(\xi,A)$ is selfadjoint on
\begin{equation*}
\cD(H_0(\xi,A))=\cD(d\Gamma_A(\omega))\cap \cD(K(\xi-d\Gamma_A(m)))
\end{equation*}
and $H_\mu(\xi,A)\geq -\mu^2\lVert \omega^{-1/2}v\lVert$. Furthermore, $H_\mu(\xi,A)$ is essentially selfadjoint on any core for $H_0(\xi,A)$.
\end{lemma}
\begin{proof}
For each $n\in \NN$ we define a map $G^{(n)}_\xi=K(\xi-m^{(n)})+\omega^{(n)}$ and let $G_\xi^{(0)}=K(\xi)$. Define $B_\xi=\bigoplus_{n=0}^\infty G^{(n)}_\xi$ on $\cF(\cH_A)$.  Using that $ K(\xi-m^{(n)}) $ and  $ \omega^{(n)}$ are non negative multiplication operators we find	
\begin{align*}
\cD(B_\xi)= \cD(K_A(\xi-d\Gamma(m)))\cap \cD(d\Gamma(\omega)) \,\, \text{and}\,\, H_0(\xi,A)=B_\xi,
\end{align*}
so $ H_0(\xi,A)$ is selfadjoint. For $\psi\in \cD(H_0(\xi,A))$ we have $\lVert d\Gamma(\omega)\psi \lVert\leq  \lVert H_0(\xi,A)\psi \lVert$ and so we find via Lemma \ref{Lem:FundamentalIneq} and the Kato Rellich theorem that
\begin{align*}
H_\mu (\xi,A):=H_0(\xi,A)+\mu \varphi(v_A)
\end{align*}
is selfadjoint on $\cD(H_0(\xi,A))$ and any core for $H_0(\xi,A)$ is a core for $H_\mu(\xi,A)$. Using Lemma \ref{Lem:FundamentalIneq} again we find $H_\mu (\xi,A)\geq  -\mu^2\lVert \omega^{-1/2}v\lVert$.\cqfd
\end{proof}
\noindent \textbf{Hypothesis 1:}
We assume
\begin{enumerate}
\item[\textup{(1)}]  $K\in C^2(\RR^\nu,\RR)$ is non negative and there is $C_K>0$ such that
\begin{align*}
\lvert \nabla K(x) \lvert^2 &\leq C_K(1+ K(x)^2) \,\,\,\, \text{and}\\
\lVert D^2K(x) \lVert &\leq 2C_K
\end{align*}
for all $x\in \RR^\nu$. Here $D^2K$ is the Hessian of $K$ and $\lVert D^2K(x) \lVert$ is the operator norm of $D^2K(x)$.

\item[\textup{(2)}] $\omega:\RR^\nu\rightarrow [0,\infty)$ is continuous and $\omega>0$ almost everywhere.

\item[\textup{(3)}] $v\in \cD(\omega^{-1/2})$.
\end{enumerate}
Under these assumptions we may define
\begin{align*}
\nabla K(\xi-d\Gamma_A(m))&=(\partial_1K(\xi-d\Gamma_A(m)),\dots,\partial_{\nu}K(\xi-d\Gamma_A(m)) )\\ \Sigma_A(\xi)&=\inf(\sigma(H_\mu(\xi,A)))
\end{align*}
We have the following lemma
\begin{lemma}\label{Lem:FundamentalTechnicalStuff}
Assume Hypothesis 1. The following holds
\begin{enumerate}
\item[\textup{(1)}] $\cD(K(\xi-d\Gamma_A(m))) \subset \cD (\nabla K(\xi-d\Gamma_A(m)))$ and
\begin{align*}
\lVert \nabla K(\xi-d\Gamma_A(m))\psi \lVert^2\leq C_K\lVert K (\xi-d\Gamma_A(m))\psi\lVert^2+C_K\lVert \psi\lVert^2
\end{align*}
for all $\psi\in \cD(K(\xi-d\Gamma_A(m)))$.

\item[\textup{(2)}] $\cD(K(\xi-d\Gamma_A(m)))$ is independent of $\xi$. On $\cD(K(\xi-d\Gamma_A(m)))$ we have for all $a\in \RR^\nu$
\begin{align}\label{eq:vigtigID}
K(\xi+a-d\Gamma_A(m))=&K(\xi-d\Gamma_A(m))+ a\cdot  \nabla K(\xi-d\Gamma_A(m)) \nonumber \\  &+ E_{\xi,A}(a)
\end{align}
 where $\lVert E_{\xi,A}(a) \lVert\leq C_K\lvert a\lvert^2$. In particular, $\cD(H_\mu(\xi,A))$ is independent of $\xi$.

\item[\textup{(3)}] Let $\psi \in \cD(K(\xi-d\Gamma_A(m)))$. Then
\begin{align}\nonumber
\lVert K(\xi+a&-d\Gamma_A(m))\psi-K(\xi-d\Gamma_A(m))\psi\lVert^2\\&\leq 2C_K\lvert a\lvert^2 \lVert K(\xi-d\Gamma_A(m))\psi \lVert^2+2(1+C_K\lvert a\lvert^2)C_K\lvert a\lvert^2\lVert \psi\lVert^2.\label{eq:Vigtig:ulighed}
\end{align}
Furthermore, $\xi\mapsto H_\mu(\xi,A)\psi$ is continuous for any $\psi\in \cD(H_\mu(0,A))$ and $\xi\mapsto H_\mu(\xi,A)$ is continuous in norm resolvent sense. In particular, the map $\xi\mapsto \Sigma_A(\xi)$ is continuous.

\item[\textup{(4)}] Let $\cD\subset \cC\cS_A$ be a dense subspace. Then $\cE(\cD)$ and $\cJ(\cD)$ span cores for $H_\mu(\xi,A)$.
\end{enumerate}
\end{lemma}
\begin{proof}
To prove (1) we calculate for $\psi\in \cD(K(\xi-d\Gamma_A(m)))$
\begin{align*}
\sum_{i=1}^{\nu}&\sum_{n=0}^{\infty}\int_{A^n}\lvert \psi^{(n)}(k) \partial_i K(\xi-m^{(n)}(k))\lvert^2d\lambda_{n\nu}\\&\leq \sum_{n=0}^{\infty}\int_{A^n}C_K\lvert \psi^{(n)}(k)  K(\xi-m^{(n)}(k))\lvert^2d\lambda_{n\nu}+C_K\lVert \psi\lVert^2\\&=C_K\lVert K (\xi-d\Gamma_A(m))\psi\lVert^2+C_K\lVert \psi\lVert^2.
\end{align*}
This proves $(1)$. To prove $(2)$ we use the fundamental theorem of calculus twice we find for all $k,\xi\in \RR^\nu$
\begin{align}\label{eq:fisk}
K(\xi+a-k)=&K(\xi-k)+ a\cdot \nabla K(\xi-k)\nonumber \\&+ \int_0^1\int_0^1 ta\cdot D^2K(\xi+sta-k)a dsdt    
\end{align}
Define
\begin{equation*}
G_{a}(x)=\int_0^1\int_0^1 ta\cdot D^2K(x+sta)a dsdt    
\end{equation*}
and $E_{\xi,A}(a)=G_{a}(\xi-d\Gamma_A(m))$. Then $E_{\xi,A}(a)$ is bounded with norm bound $C_K\lvert a\lvert^2$ since $\lvert G_a(x)\lvert\leq C_K \lvert a\lvert^2$ uniformly in $x$. Let $\psi\in \cD(K(\xi-d\Gamma_A(m)))$. Then $\psi\in \cD(  \nabla K(\xi-d\Gamma_A(m)))$ by part (1) and equation (\ref{eq:fisk}) gives
\begin{align*}
((K(\xi-d\Gamma_A(m))&+ a\cdot  \nabla K(\xi-d\Gamma_A(m))+ E_{\xi,A}(a))\psi)^{(n)}\\&=K(\xi+a-m^{(n)})\psi^{(n)}
\end{align*}
pointwise showing $\psi \in \cD(K(\xi+a-d\Gamma_A(m)))$ and equation (\ref{eq:vigtigID}) holds. We have thus proven $\cD(K(\xi+a-d\Gamma_A(m)))\subset \cD(K(\xi-d\Gamma_A(m)))$ for all $\xi\in  \RR^\nu$. Using $\xi'=\xi-a$ we find the other inclusion. This proves (2).

To prove (3) we note that equation (\ref{eq:Vigtig:ulighed}) is easily obatined from statements (1) and (2). Using
\begin{align*}
(H_\mu(\xi+a,A)-H_\mu(\xi,A))\psi=(K(\xi+a-d\Gamma_A(m))-K(\xi-d\Gamma_A(m)))\psi
\end{align*}
for any $\psi\in \cD(H_\mu(\xi,A))$ and equation (\ref{eq:Vigtig:ulighed}) we immediately obtain continuity of $\xi \mapsto H_\mu(\xi,A)\psi$. To prove the statement regarding norm resolvent convergence we calculate using equation (\ref{eq:Vigtig:ulighed})
\begin{align*}
\lVert&(H_\mu(\xi+a,A)+i)^{-1}-(H_\mu(\xi,A)+i)^{-1}\lVert^2\\&\leq 2C_K\lvert a\lvert^2\lVert K_A(\xi-d\Gamma(m))(H_\mu(\xi,A)+i)^{-1} \lVert+(1+C_K\lvert a\lvert^2)C_K\lvert a\lvert^2
\end{align*}
which goes to 0 for $a$ tending to 0. Continuity of $\xi\mapsto \inf(\sigma(H_\mu(\xi,A)))$ now follows from continuity of the spectral calculus and the existence of a $\xi$-independent lower bound by Lemma \ref{Lem:BasicProperties} (see e.g. \cite[Lemma 5.5]{SB2}).

It only remains to prove statement (4). By Lemma \ref{Lem:BasicProperties} it is enough to check that $\cJ(\cD)$ and $\cE(\cD)$ span a core for $H_0(\xi,A)$. Let $f_1,\dots,f_n\in \cC\cS_A$. Pick $R>0$ such that $1_{B_{R}(0)}f_i=f_i$ almost everywhere on $A$ for all $i\in \{1,\dots,n\}$ and note that $1_{B_{R}(0)^n}f_1\otimes_s\dots\otimes_s f_n=f_1\otimes_s\dots\otimes_s f_n$ almost everywhere on $A^n$. Let $C=\sup_{k\in B_R(0)}\omega(k)$. Using equation $(\ref{eq:fisk})$ we find the following point wise inequality for $k\in B_{R}(0)^n:$
\begin{align*}
\lvert K(\xi-m^{(n)}(k))\lvert =K(\xi)+\lvert m^{(n)}(k)\lvert \lVert \nabla K(\xi)\lvert+ \lvert m^{(n)}(k)\lvert^2 C_K\leq \widetilde{C}(1+n^2R^2)
\end{align*}
Where $\widetilde{C}=\max\{ K(\xi)+\frac{1}{2}\lvert \nabla K(\xi)\lvert,(1+C_K)    \}$ and we used that $\lvert m^{(n)}(k)\lvert\leq nR$ for $k\in B_{R}(0)^n$. This implies $ f_1\otimes_s\dots\otimes_s f_n\in \cD(H_0(\xi,A)^p)$ for all $p>0$ and
\begin{align}\label{eq:beregning}
\lVert H_0(\xi,A)^p f_1\otimes_s\dots\otimes_s f_n \lVert\leq  (\widetilde{C}(1+n^2R^2)+nC)^{p}\lVert f_1\otimes_s\dots\otimes_s f_n \lVert
\end{align}
Multiplying by $\frac{1}{p!}$ and summing over $p$ yields a finite number so $f_1\otimes_s\dots\otimes_s f_n$ is analytic for $H_0(\xi,A)$. Furthermore, $\Omega$ is an eigenvector for $H_0(\xi,A)$ and therefore analytic. In conclusion, $\cJ(\cD)$ is a total set of analytic vectors for $H_0(\xi,A)$ and therefore it spans a core for $H_0(\xi,A)$ by Nelsons analytic vector theorem.

By equation (\ref{eq:beregning}) we see $f^{\otimes n}\in  \cD(H_0(\xi,A)^p)$ and 
\begin{align*}
\lVert H_0(\xi,A)^p f_1^{\otimes n} \lVert^2&\leq \lVert (\widetilde{C}(1+n^2R^2)+nC)^{2p}\lVert f_1\lVert^{2n}\\&\leq (\widetilde{C}^{1/2}(1+nR)+\sqrt{nC})^{4p}\lVert f_1\lVert^{2n}
\end{align*}
This also holds for $n=0$ as we in this case obtain $ \lVert H_0(\xi,A)^p\Omega\lVert^2=\sqrt{K(\xi)}^{4p}\leq (\widetilde{C}^{1/2})^{4p}$. Multiplying by $\frac{1}{n!}$ and summing over $n$ yields a finitie number so $ \epsilon(f_1)\in \cD(H_0(\xi,A)^p)$ for all $p>0$. Now
\begin{align*}
\sum_{p=0}^{\infty}\frac{1}{(2p)!}\lVert H_0(\xi,A)^p\epsilon(f_1) \lVert &\leq \sum_{p=0}^{\infty}\frac{1}{(2p)!}\sum_{n=0}^\infty \frac{1}{\sqrt{n!}}\lVert H_0(\xi,A)^pf_1^{\otimes n}\lVert \\&=\sum_{n=0}^\infty\frac{1}{\sqrt{n!}} \sum_{p=0}^{\infty}\frac{1}{(2p)!}(\widetilde{C}^{1/2}(1+nR)+\sqrt{nC})^{2p}\lVert f_1\lVert^{n}\\&\leq \sum_{n=0}^\infty\frac{\lVert f\lVert^ne^{(\widetilde{C}^{1/2}(1+nR)+\sqrt{nC})}}{\sqrt{n!}}<\infty
\end{align*}
Thus $\varepsilon(f_1)$ is semi analytic for $H_0(\xi)$. This implies $\cE(\cD)$ spans a dense subspace of semi analytic vectors for $H_0(\xi,A)$, so $\cE(\cD)$ spans a core for $H_0(\xi,A)$ by the Masson-McClary theorem. \cqfd
\end{proof}
\textbf{Hypothesis 2:}
Hypothesis 2 is said to hold if
\begin{enumerate}
\item[\textup{(1)}]  $K$ is rotation invariant and $k\mapsto e^{-tK(k)}$ is positive definite for all $t\in [0,\infty)$.

\item[\textup{(2)}] $\omega$ is rotation invariant, sub-additive, $\omega(x_1)< \omega(x_2)$ if $\lvert x_1\lvert< \lvert x_2\lvert$ and
\begin{equation*}
C_\omega=\lim_{k\rightarrow 0} \lvert k\lvert^{-1} \omega(k)
\end{equation*}
exists and is strictly positive.

\item[\textup{(3)}] $v$ is rotation invariant and $\omega^{-1}v\notin \cH$.
\end{enumerate}
Our main result is
\begin{theorem}\label{Thm:Mainthm}
Assume $\nu\geq 3$ and that Hypotheses 1 and 2 holds. Then $H_\mu(\xi)$ has no ground states for any $\xi\in \RR^\nu$ and $\mu\neq 0$.
\end{theorem}
The physical choices for the 3-dimensional Nelson model are $\omega(k)=\lvert k\lvert $, $K\in\{k\mapsto\lvert k\lvert^2,k\mapsto \sqrt{\lvert k\lvert^2+m}-m \}$ and $v=\omega^{-1/2}\chi$ where $\chi:\RR^\nu \rightarrow \RR$ is a spherically symmetric ultraviolet cutoff. It is well known that Hypothesis 1 and 2 are fulfilled in this case so our result holds. 

\section{Proof of Theorem \ref{Thm:Mainthm}}
We start with proving series of lemmas which we shall need. The first lemma is well known but we provide a proof to ensure completeness of the presentation.
\begin{lemma}\label{Lem:rotinv}
Assume Hypotheses 1 and 2 holds. Then $\xi\mapsto \Sigma(\xi)$ is rotation invariant. 
\end{lemma}
\begin{proof}Let $O$ denote any orthogonal matrix with dimensions $\nu$. Define the unitary map $\widehat{O}:\cH\rightarrow \cH$ by $(\widehat{O}f)(k)=f(Ok)$ almost everywhere. Let $f,g\in \cC\cS$ and note $\widehat{O}f,\widehat{O}g\in \cC\cS$. In particular, $\Gamma(\widehat{O})\epsilon(f)=\epsilon(\widehat{O}f)\in \cD(H_\mu(\xi))$ for all $\xi\in \RR^\nu$. One now easily calculates using Lemma \ref{Lem:FundamentalCalculations}
\begin{align*}
\langle \epsilon(g), \Gamma(\widehat{O})^*H_\mu(\xi)\Gamma(\widehat{O})\epsilon(f)\rangle&=\langle \epsilon(\widehat{O}g), H_\mu(\xi)\epsilon(\widehat{O}f)\rangle\\&=\langle \epsilon(g), H_\mu(O\xi)\epsilon(f)\rangle.
\end{align*}
Using $\cE(\cC\cS)$ is total so we find $H_\mu(O\xi)=\Gamma(\widehat{O})^*H_\mu(\xi)\Gamma(\widehat{O})$ on $\cE(\cC\cS)$ which spans a core for $H_\mu(O\xi)$ and so $\Gamma(\widehat{O})^*H_\mu(\xi)\Gamma(\widehat{O})=H_\mu(O\xi)$. \cqfd
\end{proof}
For any $k\in \RR^\nu \backslash \{ 0 \}$ we write $\widehat{k}=\lvert k\lvert^{-1}k$. The next lemma is essentially spherical coordinates but is proven here for completeness of the presentation

\begin{lemma}\label{Lem:Globmin}
Assume Hypotheses 1 and 2 holds. Then $\Sigma$ has a global minimum at $\xi=0$. 
\end{lemma}
\begin{proof}
The theorem is proven in \cite{Gross} under the extra assumptions that $\omega(k)=\sqrt{\lvert k\lvert^2+m^2}$ and $v(k)=g\omega(k)^{-1/2}1_{B_R(0)}(k)$ for some $m>0$ and $R>0$. To extend the result let $\cK\subset\cH$ denote the set of real valued functions and $F_2$ denote the Fourier transform on $\cH$ normalized so $F_2$ is unitary. The proof presented in paper \cite{Gross} only uses
\begin{enumerate}
	\item[\textup{(1)}] $v\in F_2(\cK)$,
	\item[\textup{(2)}] $e^{-t\omega}F_2(\cK)\subset F_2(\cK)$ to ensure $e^{-td\Gamma(\omega)}$ is positivity preserving on the Q-space generated by $F_2(\cK)$ (see e.g. \cite{Hirokawa1}).
\end{enumerate}
Assume first $v$ is real valued. Then part (1) follows from rotation invariance of $v$. To prove part (2), let $U:\cH\rightarrow \cH $ be the unitary map defined by $(Uf)(k)=f(-k)$ almost everywhere and let $\cC:\cH\rightarrow \cH$ be the anti linear map defined by $(\cC f)(k)=\overline{f(k)}$. For $\psi\in \cK$ we have
\begin{align*}
 e^{-t\omega} F_2\psi=  F_2   F_2^* e^{-t\omega} F_2 \psi.
\end{align*}
Note that $F_2U=F_2^*=UF_2$ and $\cC U F_2 \cC=F_2=\cC  F_2 U \cC$, so
\begin{align*}
\cC F_2^* e^{-t\omega} F_2 \psi=F_2 \cC e^{-t\omega} \cC U F_2 \cC \psi =F_2U e^{-t\omega}F_2\psi=F_2^* e^{-t\omega} F_2 \psi,
\end{align*}
since $\omega(k)=\omega(-k)$ almost everywhere. This proves $F_2^* e^{-t\omega} F_2 \psi\in \cK$ finishing the proof of the lemma when $v$ is real valued. For general $v$, we let $\phi:\RR^\nu\rightarrow \CC$ be a measurable function such that $\lvert \phi (k)\lvert=1$ for all $k\in \RR^\nu$ and $\phi(k) \lvert v (k)\lvert= v(k)$ almost everywhere. Then $\phi$ defines a unitary multiplication operator and so $\Gamma(\phi)$ is unitary. A calculation similar to the one in Lemma \ref{Lem:rotinv} shows
\begin{align*}
 \Gamma(\phi)^*H_\mu(\xi)\Gamma(\phi)= K(\xi-d\Gamma(m))+d\Gamma(\omega)+\mu \varphi(\lvert v\lvert)
\end{align*}
so we may assume $v$ is real valued finishing the proof. 
\cqfd
\end{proof}
For every $\xi\in \RR^\nu$ and $0<\varepsilon<1$ we define
\begin{align*}
S_{\varepsilon}(\xi)&=\{ k\in \RR^{\nu}\backslash \{0\} \mid 2\lvert \widehat{k}\cdot \xi\lvert< (1-\varepsilon)\lvert \xi\lvert  \}.
\end{align*}
\begin{lemma}\label{Lem:Energyineq}
Assume Hypotheses 1 and 2 holds. Let $\xi\in \RR^\nu$. Then
\begin{enumerate}
\item[\textup{(1)}] $\Sigma(\xi-k)+\omega(k)>\Sigma(\xi)$ if $k\notin \RR \xi$.

\item[\textup{(2)}] For any $\varepsilon\in (0,1)$ there exists $D:=D(\varepsilon,\xi)<1$ and $r:=r(\varepsilon,\xi)>0$ such that for all $k\in B_{r}(0)\cap S_{\varepsilon}(\xi)$ we have
\begin{align*}
\Sigma(\xi-k)-\Sigma(\xi)\geq -D\omega(k) 
\end{align*}
\end{enumerate}
\end{lemma}
\begin{proof}
First we note that $\omega(k)>0$ if $k\neq 0$ because if $\omega(k)=0$ then $\omega(k')<0$ for all $k'\in B_{\lvert k\lvert}(0)$ by Hypothesis 2 which contradicts Hypothesis 1. Using this observation and Lemma \ref{Lem:Globmin} we find 
\begin{align*}
\Sigma(0-k)-\Sigma(0)\geq 0 > -\frac{1}{2}\omega(k)>-\omega(k)
\end{align*}
for all $k\neq 0$. This proves the statement when $\xi=0$. Assume $\xi\neq 0$ and let $k\notin \RR\xi$. By rotation invariance of $\Sigma$ (Lemma \ref{Lem:rotinv}) we may calculate
\begin{align}\label{eq:regn1}
\Sigma(\xi-k)-\Sigma(\xi)=\Sigma(\xi-k)-\Sigma\left( \frac{\lvert \xi\lvert }{\lvert \xi-k\lvert} (\xi-k) \right)
\end{align}
By Lemma \ref{Lem:HVZ} we have $\Sigma(\xi-k)+\omega(k)\in \sigma(H_\mu(\xi))$ and so $\Sigma(\xi)\leq \Sigma(\xi-k)+\omega(k)$ implying  
\begin{align}\nonumber
\Sigma(\xi-k)-\Sigma\left( \frac{\lvert \xi\lvert }{\lvert \xi-k\lvert} (\xi-k) \right)&\geq -\omega\left(   \frac{\lvert \xi\lvert }{\lvert \xi-k\lvert} (\xi-k)-\xi-k \right)\\&=-\omega\left(  (\lvert \xi\lvert-\lvert \xi-k\lvert)\frac{\xi-k}{\lvert \xi-k\lvert} \right)\label{eq:regn2}
\end{align}
Now $\lvert \lvert \xi\lvert-\lvert \xi-k\lvert\lvert \leq \lvert k \lvert$ by the reverse triangle inequality. If equality holds then $\lvert \xi\lvert=\lvert \xi-k\lvert+\lvert k\lvert$ or $\lvert \xi-k\lvert=\lvert \xi\lvert+\lvert -k\lvert$. By \cite[Page 9]{Weidmann} either $k$ and $\xi-k$ are linearly dependent or $k$ and $\xi$ are linearly dependent. In any case we have a contradiction with $k\notin \RR\xi$. So $\lvert \lvert \xi\lvert-\lvert \xi-k\lvert\lvert < \lvert k \lvert$ which implies
\begin{align*}
\omega\left(  (\lvert \xi\lvert-\lvert \xi-k\lvert)\frac{\xi-k}{\lvert \xi-k\lvert} \right)<\omega(k)
\end{align*}
by Hypothesis 2. Combining this and equations (\ref{eq:regn1}) and (\ref{eq:regn2}) we find statement (1). To prove statement (2) we continue to calculate for $k\in S_\varepsilon(\xi)$ (which is disjoint from $\RR \xi$)
\begin{align}\label{eq:regn3}
\lvert \lvert \xi-k\lvert-\lvert \xi\lvert\lvert &=\left \lvert \frac{\lvert \xi-k\lvert^2-\lvert \xi\lvert^2}{\lvert \xi-k\lvert+\lvert \xi\lvert}\right \lvert\nonumber \\&=\lvert k\lvert \left \lvert \frac{-2\xi \cdot \widehat{k}+\lvert k \lvert }{\lvert \xi-k\lvert+\lvert \xi\lvert}\right \lvert \\&\leq  \lvert k\lvert \left(1-\varepsilon +\frac{\lvert k \lvert}{\lvert \xi \lvert } \right)\nonumber 
\end{align}
Pick $n$ such that $D:=(1+1/n)(1-1/n)^{-1}(1-\varepsilon/2)<1$ and $R>0$ such that
\begin{equation}\label{eq:regn4}
C_\omega(1-1/n)\lvert k \lvert \leq \omega(k)\leq C_\omega(1+1/n)\lvert k \lvert
\end{equation}
for all $k\in B_R(0)$. Pick $r=\min\{\frac{\lvert \xi\lvert\varepsilon }{2},R\}$. Using equations (\ref{eq:regn1}), (\ref{eq:regn2}), (\ref{eq:regn3}) and (\ref{eq:regn4}) we find
\begin{align*}
\Sigma(\xi-k)-\Sigma(\xi)&\geq -C_\omega(1+1/n)\left(1-\varepsilon +\frac{\lvert k \lvert}{\lvert \xi \lvert } \right)\lvert k\lvert \geq - D\omega(k)
\end{align*}
for $k\in B_{r}(0)\cap S_{\varepsilon}(\xi)$. \cqfd
\end{proof}
The following lemma is well known see e.g. \cite[Lemma 6.2]{SB2}.
\begin{lemma}\label{Lem:Uniquenessimpel}
Define $A=\{v\neq 0\}$ and assume Hypotheses 1 and 2 holds. If  $H_\mu(\xi,A)$ has a ground state for some $\mu\neq 0$ and $\xi\in \RR^\nu$ then the corresponding eigenspace is non degenerate.
\end{lemma}
\begin{lemma}\label{Lem:Uniqueness}
Assume $\nu \geq 2$ and Hypotheses 1 and 2 holds. If $H_\mu(\xi)$ has a ground state for some $\mu\neq 0$ and $\xi\in \RR^\nu$ then the corresponding eigenspace is non degenerate.
\end{lemma}
\begin{proof}
Define $A=\{v\neq 0\}$. By Lemma \ref{Lem:Splitspace} there is a unitary map
\begin{equation*}
U:\mathcal{F}(\mathcal{H})\rightarrow \mathcal{F}(\mathcal{H}_A)\oplus \bigoplus_{n=1}^{\infty} \mathcal{F}(\mathcal{H}_A)\otimes \mathcal{H}_{A^c}^{\otimes_s n}
\end{equation*}
such that 
\begin{align}\label{eq:ulighed}
U H_\mu(\xi)U^*=H_\mu(\xi,A)\oplus \bigoplus_{n=1}^\infty H_{n,\mu}(\xi,A)\mid_{\mathcal{F}(\mathcal{H}_A)\otimes \mathcal{H}_{A^c}^{\otimes_s n}}
\end{align}
for all $\xi \in \RR^\nu$ where
\begin{align*}
H_{n,\mu}(\xi,A)=\int^{\oplus}_{(A^c)^n} H_\mu(\xi-k_1-\dots-k_n,A)+\omega(k_1)+\dots+\omega(k_n)d\lambda_{\nu}^{\otimes n}
\end{align*}
Let $\psi$ be any ground state for $H_{\mu}(\xi)$. We prove $U\psi=(\widetilde{\psi}^{(0)},0,0,\dots)$. Write $U\psi=(\widetilde{\psi}^{(n)})$ and assume towards contradiction that $\widetilde{\psi}^{(n)}\neq 0$ for some $n\geq 1$. Then $\widetilde{\psi}^{(n)}$ is an eigenvector for $H_{n,\mu}(\xi,A)$ corresponding to the eigenvalue $\Sigma(\xi)$. The spectral projection of $H_{n,\mu}(\xi,A)$ onto $\Sigma(\xi)$ is given by
\begin{align*}
\int^{\oplus}_{(A^c)^n} 1_{\{ \Sigma(\xi) \}}(H_\mu(\xi-k_1-\dots-k_n,A)+\omega(k_1)+\dots+\omega(k_n))d\lambda_{\nu}^{\otimes n} \neq 0.
\end{align*}
Hence $\Sigma(\xi)$ is an eigenvalue for $H_\mu (\xi-k_1-\dots-k_n,A)+\omega(k_1)+\dots+\omega(k_n)$ on a set of positive $\lambda_{n\nu}=\lambda_{\nu}^{\otimes n}$ measure. Sub-additivity of $\omega$ along with Lemmas \ref{Lem:Splitspace} and \ref{Lem:HVZ} gives
\begin{align*}
\Sigma(\xi)&\geq \Sigma_A(\xi-k_1-\dots-k_n)+\omega(k_1)+\dots+\omega(k_n)\\&\geq \Sigma(\xi-k_1-\dots-k_n)+\omega (k_1+\dots+k_n)\geq \Sigma(\xi)
\end{align*}
most hold on a set of positive $\lambda_{n\nu}$ measure. By Lemma \ref{Lem:Energyineq} we se that this can only hold for $(k_1,\dots,k_n)\in (\RR^\nu)^n$ with $k_1+\dots+k_n\in \text{Span}(\xi)$. However, the set of $k$ satisfying this is a subspace of $(\RR^\nu)^n$ of dimension $\nu (n-1)+1<\nu n$ and such a subspace must have $\lambda_{n\nu}$ measure 0 which is a contradiction.

Assume $\psi_1,\psi_2$ are orthogonal eigenvectors corresponding to the eigenvalue $\Sigma(\xi)$. Then $U\psi_i=(\widetilde{\psi_i},0,0,\dots.)$. Now $U$ preserves the inner product so $\widetilde{\psi_1}$ and $\widetilde{\psi_2}$ are orthogonal eigenvectors for $H_\mu(\xi,A)$ corresponding to the eigenvalue $\Sigma(\xi)$ so in particular $\Sigma(\xi)\geq \Sigma_A(\xi)$. By equation (\ref{eq:ulighed}) we conclude that $\Sigma(\xi)= \Sigma_A(\xi)$ and therefore $H_\mu(\xi,A)$ has two orthogonal ground states. This is a contradiction with Lemma \ref{Lem:Uniquenessimpel}.  \cqfd
\end{proof}
The next two Lemmas are an adapted version of the corresponding ones found in \cite{HerbstHasler}. For $\xi\in \RR^\nu$ and $k\neq 0$ we define
\begin{align*}
Q_0(k,\xi)&=\omega(k)(H_\mu(\xi)-\Sigma(\xi)+\omega(k))^{-1} \\
P_0(\xi)&=1_{\{\Sigma(\xi)\}}(H_\mu(\xi))
\end{align*}
\begin{lemma}\label{Lem:Q0egenskaber}
Assume Hypotheses 1 and 2 holds. Fix $\xi\in \RR^\nu$ and $R>0$. Then $\widehat{k}\cdot \nabla K(\xi-d\Gamma(m))Q_0(k,\xi)$ is uniformly bounded for $k$ in $B(0,R)\backslash \{0\}$ and
\begin{equation}
s-\lim_{k\rightarrow 0} \widehat{k}\cdot \nabla K(\xi-d\Gamma(m))Q_0(k,\xi)(1-P_0(\xi))= 0
\end{equation}
\end{lemma}
\begin{proof}
Note $ \widehat{k}\cdot\nabla K(\xi-d\Gamma(m))Q_0(k)$ is bounded for $k\neq 0$ by the closed graph theorem and Lemma \ref{Lem:FundamentalTechnicalStuff}. For $\psi\in \cF(\cH)$ we find by equation (\ref{eq:Vectorcauchy}) that
\begin{align*}
\lVert \widehat{k}\cdot \nabla K(\xi-d\Gamma(m))Q_0(k,\xi)\psi \lVert^2\leq \sum_{i=1}^{\nu}\lVert \partial_i K(\xi-d\Gamma(m))Q_0(k,\xi)\psi \lVert^2
\end{align*}
so it is enough to see $\partial_i K(\xi-d\Gamma(m))Q_0(k,\xi)$ is uniformly bounded on $B(0,R)\backslash \{0\}$ for any $R>0$ and converges strongly to $\partial_i K(\xi-d\Gamma(m))P_0(\xi)$. We have
\begin{align*}
\partial_i K(\xi-d\Gamma(m))&Q_0(k,\xi)=\partial_i K(\xi-d\Gamma(m))\frac{\omega(k)}{H_\mu(\xi)-\Sigma(\xi)+\omega(k)+1} \\&+\partial_i K(\xi-d\Gamma(m))\frac{1}{H_\mu(\xi)-\Sigma(\xi)+\omega(k)+1}Q_0(k,\xi)
\end{align*}
Now $\omega$ is continuous and goes to 0 as $k$ tends to 0, so $Q_0(k,\xi)$ goes strongly to $P_0(\xi)$. Hence it is enough to see $\partial_i K(\xi-d\Gamma(m))(H_\mu(\xi)-\Sigma(\xi)+\omega(k)+1)^{-1}$ is uniformly bounded in $k$ and converges to $\partial_i K(\xi-d\Gamma(m))(H_\mu(\xi)-\Sigma(\xi)+1)^{-1}$ in norm. But this follows from the equality
\begin{align*}
\partial_i K(\xi&-d\Gamma(m))\frac{1}{H_\mu(\xi)-\Sigma(\xi)+\omega(k)+1}-\partial_i K (\xi-d\Gamma(m))\frac{1}{H_\mu(\xi)-\Sigma(\xi)+1}\\&=\partial_i K(\xi-d\Gamma(m))\frac{1}{H_\mu(\xi)-\Sigma(\xi)+1}\frac{\omega(k)}{H_\mu(\xi)-\Sigma(\xi)+1+\omega(k)}
\end{align*}
\cqfd
\end{proof}
For $\xi \in \RR^\nu$ and $k\notin \RR\xi$ we may by Lemma \ref{Lem:Energyineq} define
\begin{align*}
Q(k,\xi)=\omega(k)(H_\mu(\xi-k)-\Sigma(\xi)+\omega(k))^{-1}.
\end{align*}
\begin{lemma}\label{Lem:Mainconv}
Fix $\xi\in \RR^\nu$ and assume Hypotheses 1 and 2 holds. There is a vector $u(\xi)\in \RR^\nu$ such that
\begin{align*}
P_0(\xi)\widehat{k}\cdot \nabla K(\xi-d\Gamma(m)) P_0(\xi)=\widehat{k}\cdot u(\xi)P_0(\xi)
\end{align*}
for any $k\in \RR^\nu \backslash \{0\}$. Pick $\varepsilon\in (0,1)$ such that $\widehat{k}\cdot C_\omega u(\xi)<\frac{1}{2}$ for all $k\in S_\varepsilon(C_\omega u(\xi))$. Define
\begin{align*}
\widetilde{S}_\varepsilon(\xi)&=S_\varepsilon(\xi)\cap S_\varepsilon(C_\omega u(\xi)).
\end{align*}
If $\nu\geq 3$ then $S_\varepsilon$ is open, non-empty and invariant under positive scalings. Furthermore,
\begin{equation}\label{eq:limit}
w-\lim_{k\rightarrow 0, k\in \widetilde{S}_\varepsilon(\xi)} Q(k,\xi)-(1-C_\omega\widehat{k}\cdot u(\xi))^{-1}P_0(\xi)=0.
\end{equation}
\end{lemma}
\begin{proof}
As $\xi$ is fixed in this proof it will be omitted from the notation of $Q,Q_0, u$ and $P_0$. If $P_0=0$ we can pick $u=0$. If $P_0\neq 0$ then it has dimension 1 by Lemma \ref{Lem:Uniqueness} and is spanned by a vector $\psi\in \cD(H_\mu(\xi))$. Using $P_0=\lvert \psi\rangle \langle \psi \lvert$ we find that $u=\langle  \psi, \nabla K(\xi-d\Gamma(m)) \psi \rangle$ does the trick. Furthermore, $S_\varepsilon$ is obviously open and invariant under positive scaling since this holds for  $S_\varepsilon(\xi)$ and $ S_\varepsilon(C_\omega u)$. Furthermore, any non-zero vector which is orthogonal to $\xi$ and $u$ is in $\widetilde{S}_\varepsilon$ and such vector will always exist if $\nu\geq 3$. 

It remains only to prove equation (\ref{eq:limit}). By Lemma \ref{Lem:Energyineq} we may pick $R>0$ such that for $k\in \widetilde{S}_\varepsilon(\xi)\cap B_{R}(0)$ we have
\begin{align*}
\Sigma(k-\xi)-\Sigma(\xi)+\omega(k)\geq (1-D)\omega(k)
\end{align*}
with $D<1$. Hence we find
\begin{align}\label{eq:uniformbound1}
\lVert Q(k) \lVert&\leq (1-D)^{-1}\,\,\,\, \forall k\in \widetilde{S}_\varepsilon(\xi) \cap B_{R}(0)
\end{align}
Using Lemma \ref{Lem:FundamentalTechnicalStuff} we may calculate for $k\in \widetilde{S}_\varepsilon(\xi)$:
\begin{align}
Q(k)&=Q_0(k)+\frac{\lvert k\lvert}{\omega(k)}Q_0(k)(\widehat{k}\cdot \nabla K(\xi-d\Gamma(m))   )Q(k)+o(k)\label{eq:Hojre}
\end{align}
where $o_1(k):=-Q_0(k)\omega(k)^{-1}E_\xi(-k)Q(k)$.  Note $o(k)$ goes to 0 in norm for $k$ tending to 0 in $\widetilde{S}_\varepsilon(\xi)$ by equation (\ref{eq:uniformbound1}), Lemma \ref{Lem:FundamentalTechnicalStuff} and the uniform bound $\lVert Q_0(k) \lVert\leq 1$. It follows from Lemma $\ref{Lem:Q0egenskaber}$ that $(1-P_0)Q(k)$ converges weakly to 0 for $k$ tending to 0 inside $\widetilde{S}_\varepsilon(\xi)$  0 inside $\widetilde{S}_\varepsilon(\xi)$. Taking adjoints we find the same conclusion for $Q(k)(1-P_0)$. Hence we find
\begin{equation}\label{eq:limweak}
w-\lim_{k\rightarrow 0, k\in \widetilde{S}_\varepsilon(\xi)} Q(k)-P_0Q(k)P_0=0.
\end{equation}
From equation (\ref{eq:Hojre}) we find
\begin{align*}
P_0Q(k)P_0&=P_0Q_0(k)P_0+\frac{\lvert k\lvert}{\omega(k)}P_0Q_0(k)(k\cdot \nabla K(\xi-d\Gamma(m))Q(k)P_0+P_0o_1(k)P_{0}\\&=P_0+\frac{\lvert k\lvert}{\omega(k)}P_0(\widehat{k}\cdot \nabla K(\xi-d\Gamma(m)))(1-P_0)Q(k)P_0\\&+\left(\frac{\lvert k\lvert}{\omega(k)}-C_\omega \right)\widehat{k}\cdot uP_0Q(k)P_0+C_\omega\widehat{k}\cdot uP_0Q(k)P_0+P_0o_1(k)P_{0}
\end{align*}
Write $D_k=(1-C_\omega \widehat{k}\cdot u)^{-1}$. For $k\in\widetilde{S}_\varepsilon(\xi)$ we have $\lvert D_k \lvert\leq 2$ and
\begin{align*}
P_0Q(k)P_0-D_kP_0&=D_k\frac{\lvert k\lvert}{\omega(k)}P_0(\widehat{k}\cdot \nabla K(\xi-d\Gamma(m)))(1-P_0)Q(k)P_0\\&+D_k\left(\frac{\lvert k\lvert}{\omega(k)}-C_\omega\right)\widehat{k}\cdot uP_0Q(k)P_0+ D_kP_0o_1(k)P_{0}.
\end{align*}
The second and third term clearly converges to 0 in norm for $k$ tending to 0 inside $\widetilde{S}_\varepsilon(\xi)$. Sandwiching the first term with two vectors $\phi,\psi\in \cF(\cH)$ we find
\begin{align*}
D_k\frac{\lvert k\lvert}{\omega(k)}\sum_{i=1}^{n}\widehat{k}_i\langle \partial_iK(\xi-d\Gamma(m)) P_0\psi ,(1-P_0)Q(k)P_0\phi \rangle.
\end{align*}
Now $\langle \partial_iK(\xi-d\Gamma(m)) P_0\psi ,(1-P_0)Q(k)P_0\phi \rangle$ converges to 0 for $k$ going to 0 inside $\widetilde{S}_\varepsilon(\xi)$ by equation (\ref{eq:limweak}) and $\frac{\lvert k\lvert}{\omega(k)}\widehat{k}_i$ remains bounded as $k$ goes to 0. Therefore first term goes weakly to 0 for $k$ going to 0 inside $\widetilde{S}_\varepsilon(\xi)$. \cqfd
\end{proof}

\begin{proof}[Theorem \ref{Thm:Mainthm}]
Assume towards contradiction that a ground state $\psi=(\psi^{(n)})$ exists. For $k\in\RR^\nu$ we define the pointwise annihilation operator $a(k)\psi=(a(k)\psi^{(n+1)})$ where
\begin{equation*}
a(k)\psi^{(n+1)} = \sqrt{n+1}\psi^{(n+1)}(k,\cdots)\in \cH^{\otimes_s n}
\end{equation*}
for almost all $k\in\RR^\nu$. By Lemma \ref{Lem:pullthr} we get $a(k)\psi\in \cF(\cH)$ for almost all $k\in \RR^\nu$ and \begin{align*}
\langle \eta,a(k)\psi(k)\rangle=\mu \frac{v(k)}{\omega(k)}\langle \eta,Q(k)\psi\rangle.
\end{align*}
From Lemma \ref{Lem:Mainconv} we get
\begin{align*}
\lim_{k\rightarrow 0, k\in \widetilde{S}_\varepsilon(\xi)}\langle \eta,Q(k)\psi\rangle-(1-C_\omega\widehat{k}\cdot v(\xi))^{-1} \langle \eta,\psi \rangle=0.
\end{align*}
Since $(1-C_\omega\widehat{k}\cdot v(\xi))^{-1} \langle \eta,\psi \rangle$ is uniformly bounded from below in $\widetilde{S}_\varepsilon(\xi)$ by $\frac{1}{2}$ we find that there is $R>0$ such that
\begin{align*}
\lvert \langle \eta,a(k)\psi\rangle\lvert^2\geq \frac{\mu^2}{16}\frac{\lvert v(k)\lvert^2 }{\omega(k)^2}
\end{align*}
for all $k\in \widetilde{S}_\varepsilon(\xi)\cap B_R(0)$. Using Hypotheses 1 and 2 we see $\omega(Re_1)^2> 0$ because if that was not true then $\omega\leq 0$ on $ B_R(0)$ which is a contradiction. Hence we find
\begin{align*}
\infty=\int_{\RR^\nu}\frac{\lvert v(k)\lvert^2 }{\omega(k)^2}d\lambda_\nu\leq \frac{1}{\omega(Re_1)^2} \int_{B_R(0)^c}\lvert v(k)\lvert^2 d\lambda_\nu+\int_{B_R(0)}\frac{\lvert v(k)\lvert^2 }{\omega(k)^2}d\lambda_\nu.
\end{align*}
$v$ is square integrable so the integral of $\omega(k)^{-2}\lvert v(k)\lvert^2 $ over $B_R(0)$ must be infinite. Using Lemma \ref{Lem:rotint} below we find
\begin{align*}
\infty=\int_{B_R(0)}\frac{\lvert v(k)\lvert^2 }{\omega(k)^2}d\lambda_\nu=\lambda_\nu(B_1(0))\int_{0}^\infty 1_{B_R(0)}(ke_1)\frac{\lvert v(ke_1)\lvert^2 }{\omega(ke_1)^2}k^{\nu-1}d\lambda_1.
\end{align*}
As $\lambda_\nu(B_1(0))<\infty$ we see that the latter integral must be infinite. Furthermore, since $\widetilde{S}_\varepsilon(\xi)$ is open and not empty we have
\begin{align*}
\int_{\widetilde{S}_\varepsilon(\xi)\cap B_R(0)}\frac{\lvert v(k)\lvert^2 }{\omega(k)^2}d\lambda_\nu&=\nu \lambda_\nu(\widetilde{S}_\varepsilon(\xi) \cap B_1(0))\int_{0}^\infty 1_{B_R(0)}(xe_1)\frac{\lvert v(ke_1)\lvert^2 }{\omega(ke_1)^2}k^{\nu-1}d\lambda_1\\&=\infty
\end{align*}
by Lemma \ref{Lem:rotint} so $\lvert \langle \eta,a(k)\psi\rangle\lvert^2$ is not integrable. On the other hand we find
\begin{align*}
\lvert \langle \eta,&a(k)\psi\rangle\lvert^2\leq \lVert (N+1)^{1/2}\eta\lVert^2\lVert (N+1)^{-1/2}a(k)\psi\lVert^2\\&= \lVert (N+1)^{1/2}\eta\lVert^2 \sum_{i=1}^{\infty} \int_{\RR^{(n-1)\nu}}\lvert \psi^{(n)}(k,k_1,\dots,k_{n-1}) \lvert^2d\lambda_\nu^{\otimes n-1}(k_1,\dots,k_{n-1})
\end{align*}
which is integrable with integral $ \lVert (N+1)^{1/2}\eta\lVert^2\lVert \psi\lVert^2$ by definition of the Fock space norm. This is the desired contradiction. \cqfd
\end{proof}
\begin{lemma}\label{Lem:rotint}
	Let $U\subset \RR^\nu$ be invariant under multiplication by elements in $(0,\infty)$. Then for any positive, rotation invariant, measurable map $f$ we have
	\begin{align*}
	\int_{U} f(k) d\lambda_\nu=\nu \lambda_\nu(U\cap B_1(0))\int_{0}^{\infty} f(ke_1)k^{\nu-1}d\lambda_1
	\end{align*}
	where $e_1$ is the first standard basis vector. If $U$ is open then $\lambda_\nu(U\cap B_1(0))\neq 0$. 
\end{lemma}
\begin{proof}
	Consider the map $g: \RR^\nu\rightarrow [0,\infty)$ given by $g(k)=\lvert k\lvert$. Define the transformed measure on $([0,\infty),\cB([0,\infty)))$ by
	\begin{align*}
	\mu=(1_U\lambda_\nu)\circ g^{-1}
	\end{align*}
	The transformation theorem implies
	\begin{align*}
	\mu([0,a])=\lambda_\nu(a(U\cap B_1(0)))=\nu\lambda_\nu(U\cap B_1(0))\int_{0}^{a}k^{\nu-1}d\lambda_1
	\end{align*}
	for all $a>0$. By uniqueness of measures (see \cite[chapter 5]{Shilling}) we find that $\mu$ has density $\nu \lambda_\nu(U\cap B_1(0))k^{\nu-1}$ with respect to $\lambda_1$. Using that $ f(g(k)e_1)=f(k)$ we find
	\begin{align*}
	\lambda_\nu(U\cap B_1(0))\nu\int_{0}^{\infty} f(ke_1)k^{\nu-1}d\lambda_1&=\int_{0}^{\infty} f(ke_1)d\mu=\int_{U} f(k)d\lambda_\nu
	\end{align*}
	as desired. If $U$ is not empty then $U\cap B_1(0)$ is open and non empty which implies $\lambda_\nu(U\cap B_1(0))\neq 0$. \cqfd
\end{proof}

\appendix

\section{Partitions of unity and the essential spectrum.}
In this section we prove a few technical ingredients. Hypothesis 1 will be assumed throughout this section. Define $V_A:\cH\rightarrow \cH_A\oplus \cH_{A^c}$ by
\begin{align*}
V_A(f)=(P_Af,P_{A^c}f).
\end{align*}
Then $V_A$ is unitary with $V_A^*(f,g)=f1_A+g1_{A^c}$ almost everywhere. The following Lemma can be found in e.g. \cite{Thomas1}:
\begin{lemma}\label{Iso1}
There is a unique isomorphism $U:\mathcal{F}(\cH)\rightarrow  \mathcal{F}(\mathcal{H}_A)\otimes \mathcal{F}(\mathcal{H}_{A^c})$ with the property that $U(\epsilon(f))=\epsilon(P_Af_1)\otimes \epsilon(P_{A^c}f_2)$.
\end{lemma}	
	
\noindent The following Lemma is obvious.
\begin{lemma}\label{Iso2}
	There is a unique isomorphism
	\begin{equation*}
	U:\mathcal{F}(\mathcal{H}_A)\otimes \mathcal{F}(\mathcal{H}_{A^c})\rightarrow \mathcal{F}(\mathcal{H}_A)\oplus \bigoplus_{n=1}^{\infty} \mathcal{F}(\mathcal{H}_A)\otimes \mathcal{H}_{A^c}^{\otimes_s n}
	\end{equation*}
	such that
	\begin{equation*}
	U(w \otimes \{\psi^{(n)} \}_{n=0}^\infty)=\psi^{(0)}w\oplus \bigoplus_{n=1}^{\infty} w \otimes \psi^{(n)}.
	\end{equation*}
\end{lemma}
\noindent Note that we may identify 
\begin{align*}
\mathcal{F}(\mathcal{H}_A)\otimes \mathcal{H}_{A^c}^{\otimes_s n}=(1 \otimes S_n) L^2(\RR^{n\nu},\cB(\RR^{n\nu}),1_{(A^c)^n}\lambda_{n\nu},\mathcal{F}(\mathcal{H}_A))  
\end{align*}
where $1 \otimes S_n$ acts on $L^2(\RR^{n\nu},\cB(\RR^{n\nu}),\lambda_{n\nu},\mathcal{F}(\mathcal{H}_A))$ like
\begin{align*}
(1\otimes S_nf)(k_1,\dots,k_n)=\frac{1}{n!}\sum_{\sigma\in \cS_n} f(k_{\sigma(1)},\dots,k_{\sigma(n)}).
\end{align*}
Now we define
\begin{align*}
H_{\mu,A} ^{(n)}(\xi,k_1,\dots,k_n)=H_\mu(\xi-k_1-\dots-k_n,A)+\omega(k_1)+\dots+\omega(k_n)
\end{align*}
which is strongly resolvent measurable in $(k_1,\dots,k_n)\in (A^c)^n$ since $\xi\mapsto H_\mu(\xi,A)$ is strong resolvent measurable by Lemma \ref{Lem:FundamentalTechnicalStuff}. In particular,
\begin{align*}
H_{n,\mu}(\xi,A)=\int_{( A^c)^n}^{\oplus}H^{(n)}_{\mu,A}(\xi,k_1,\dots,k_n)d\lambda_{\nu}^{\otimes n}(k_1,\dots,k_n)
\end{align*}
defines a selfadjoint operator on $L^2(\RR^{n\nu},\cB(\RR^{n\nu}),\lambda_{n\nu},\mathcal{F}(\mathcal{H}_A))$ and it is reduced by the projection $1\otimes S_n$. Combining the above observations one arrives at the following lemma. 
\begin{lemma}\label{Lem:Splitspace}
Let $A\in \cB(\RR^\nu)$ and assume $1_{A}v=v$ almost everywhere. Define $j_i:\cH_i\rightarrow \cH_A\oplus \cH_{A^c}$ for $i\in\{A,A^c\}$ by $j_A(f)=(f,0)$ and $j_{A^c}(f)=(0,f)$ and define $Q_i=V_A^*j_i$. There is a unitary map
\begin{equation*}
	U:\mathcal{F}(\mathcal{H})\rightarrow \mathcal{F}(\mathcal{H}_A)\oplus \bigoplus_{n=1}^{\infty} \mathcal{F}(\mathcal{H}_A)\otimes \mathcal{H}_{A^c}^{\otimes_s n}
\end{equation*}
such that 
\begin{align}\label{ee}
UH_\mu(\xi)U^*=H_\mu(\xi,A)\oplus \bigoplus_{n=1}^\infty H_{n,\mu}(\xi,A)\mid_{\mathcal{F}(\mathcal{H}_A)\otimes \mathcal{H}_{A^c}^{\otimes_s n}}:=G_A(\xi)
\end{align}
for all $\xi\in \RR^\nu$. In particular, $\Sigma_A(\xi)\geq \Sigma(\xi)$ for all $\xi\in \RR^\nu$. Furthermore,
\begin{align*}
U\mid_{\mathcal{F}(\mathcal{H}_A)}=\Gamma(Q_A).
\end{align*}
Let $g_1,\dots,g_n\in \cH_{A^c}$ and let $\cK\subset \cC\cS_A$ be a subspace. Define
\begin{align*}
D=&\{ Q_{A^c}g_1\otimes_s\dots\otimes_s Q_{A^c}g_n \} \\&\cup \bigcup_{b=1}^\infty \{ h_1\otimes_s\cdots \otimes_s h_b\otimes_s  Q_{A^c}g_1\otimes_s\dots\otimes_s Q_{A^c}g_n\mid h_i\in Q_A\cK \}.
\end{align*}
If $\psi\in \textup{Span}(\cJ(\cK))$ we have
\begin{align}\label{eq:SimpelTrans1}
U^*(\psi\otimes (g_1\otimes_s\dots \otimes_s g_n))&\in \textup{Span}(D).\\
\lVert (H_\mu(\xi-k)-H_\mu(\xi))\Gamma(Q_A)\psi\lVert&=\lVert (H_\mu(\xi-k,A)-H_{\mu}(\xi,A))\psi\lVert.\label{eq:SimpelTrans2}\\
\lVert (H_{\mu}(\xi)-\lambda)\Gamma(Q_A)\psi\lVert&=\lVert (H_\mu(\xi,A)-\lambda)\psi\lVert.\label{eq:SimpelTrans3}
\end{align}
for all $\lambda\in \CC$.
\end{lemma}
\begin{proof}
Define $U=U_2U_1$ where $U_1$ is defined in Lemma \ref{Iso1} and $U_2$ is defined in Lemma \ref{Iso2}. Let $f,h\in  \cC\cS$ and write for $C\in \{A,A^c \}$ $f_C=P_C(f),h_C=P_C(h)\in \cC\cS_C$.  Then
\begin{align*}
U\epsilon(f)=U_2U_1\epsilon(f_A,f_{A^c})=U_2\epsilon(f_A)\otimes \epsilon(f_{A^c})=\epsilon(f_A)\oplus \bigoplus_{n=1}^\infty \epsilon(f_A)\otimes \frac{1}{\sqrt{n!}}f_{A^c}^{\otimes n}
\end{align*}
which one may check is in $\cD(G_A(\xi))$. A long but easy calculation yields
\begin{align*}
\langle \epsilon(h) , U^*G_A(\xi)U \epsilon(f)\rangle=\langle U\epsilon(h) , G(\xi)U \epsilon(f)\rangle=\langle \epsilon(h) , H_\mu (\xi) \epsilon(f)\rangle
\end{align*}
As $\cE(\cC\cS)$ is total we find $H_\mu(\xi)=U^*G_A(\xi)U$ on $\cE(\cC\cS)$ which spans a core for $H_\mu(\xi)$. Hence $U^*G_A(\xi)U=H_\mu(\xi)$ as both operators are selfadjoint. This proves the first part of the theorem. The remaining statements except equations (\ref{eq:SimpelTrans2}) and (\ref{eq:SimpelTrans3}) can be found in \cite{Thomas1}. However equations (\ref{eq:SimpelTrans2}) and (\ref{eq:SimpelTrans3}) follows from $U\mid_{\mathcal{F}(\mathcal{H}_A)}=\Gamma(Q_A)$ and equation (\ref{ee}). \cqfd
\end{proof}

\begin{lemma}\label{Lem:HVZ1}
Let $k_1,\dots,k_\ell\in \RR^\nu$ be different. If there is $\varepsilon>0$ such that $(B_{\epsilon}(k_1)\cup\dots\cup B_{\epsilon}(k_\ell))\cap \{ v\neq 0 \}$ is a $\lambda_\nu$ null-set, then $\Sigma(\xi-k_1-\dots-k_\ell)+\omega^{(\ell)}(k_1,\dots,k_\ell)\in \sigma_{ess}(H_\mu(\xi))$.
\end{lemma}
\begin{proof}
Pick $\varepsilon>0$ such that the balls $B_{\epsilon}(k_1),\dots, B_{\epsilon}(k_\ell)$ are pairwise disjoint and $(B_{\epsilon}(k_1)\cup\dots\cup B_{\epsilon}(k_\ell))\cap \{ v\neq 0 \}$ is a $\lambda_\nu$ null-set. Define $\varepsilon_n=\frac{\varepsilon}{n}$, $B^{(i)}_n=B_{\varepsilon_n}(k_i)$, $B_n=B^{(1)}_n\cup\dots\cup B^{(\ell)}_n$, $k_0=k_1+\dots+k_\ell$, $A_n=B^{(1)}_n\times\dots\times B^{(\ell)}_n$ and let
\begin{align*}
g^{(i)}_n&=\lambda_\nu(B^{(i)}_n\backslash B^{(i)}_{n+1})^{-1/2}1_{B^{(i)}_n\backslash B^{(i)}_{n+1}}, \\ \cA_n&=\{ f\in \cC\cS \mid f1_{B_{n}^c}=f  \,\, \text{almost everywhere} \},\\\cA_\infty&=\bigcup_{n=1}^\infty \cA_{n}.
\end{align*}
Note that $\cC\cS\subset \overline{\cA_\infty}$ so $\cA_{\infty}$ is a dense subspace of $\cH$. In particular, $\cJ(\cA_\infty)$ spans a core for $H_\mu(\xi-k_0)$ by Lemma \ref{Lem:FundamentalTechnicalStuff}. For each $p\in \NN$ we may thus pick a normed vector $\psi_p\in \cJ(\cA_\infty)$ such that $\lVert (H_\mu(\xi-k_0)-\Sigma(\xi-k_0))\psi_p \lVert\leq 1/p$. By Lemma \ref{Lem:FundamentalTechnicalStuff} there is $u_1(p)$ such that
\begin{align*}
\sup_{x=(x_1,\dots,x_\ell)\in A_n}\lVert (H_\mu(\xi-x_1-\dots-x_\ell)-H_\mu(\xi-k_0))\psi_p \lVert\leq \frac{1}{p}.
\end{align*}
for all $n\geq u_1(p)$. Note now that $\psi_p$ may be written as
\begin{equation*}
\psi_p=a(p) \Omega+\sum_{i=1}^{b(p)}\sum_{j=1}^{c(p)}\alpha_{i,j}(p)f^j_1(p)\otimes_s\cdots\otimes_s f^j_i(p)
\end{equation*}
for some $a(p),b(p),c(p),\alpha_{i,j}(p)$ constants and $f^j_i(p)\in \cA_\infty$. Note that each $f^j_i(p)$ is in fact contained in some $\cA_{n(i,j,p)}$ by definition so defining $u_2(p)=\max_{i,j}\{ n(i,j,p) \}$ we see that $\psi_p \in \textup{Span}(\cJ(\cA_n))$ for any $n\geq u_2(p)$. Define $u_p$ inductively by $u_1=\max\{ u_1(1),u_2(1) \}$ and $u_{p+1}=\max\{ u_1(p+1),u_2(p+1),u_{p} \}+1$. 

To summarise we have found normed vectors $\psi_p\in \cD(H_\mu(\xi))$ and a strictly increasing sequence of numbers $\{u_p\}_{p=1}^\infty\subset \NN$ such that
\begin{enumerate}
\item[\textup{(1)}] $\lVert (H_\mu(\xi-k_0)-\Sigma(\xi-k_0))\psi_p \lVert\leq 1/p$.

\item[\textup{(2)}] $\sup_{k\in B_{\delta_p}(k_1+\dots+k_\ell)}\lVert (H_\mu(\xi-k)-H_\mu(\xi-k_0))\psi_p \lVert\leq \frac{1}{p}$.

\item[\textup{(3)}] $\psi_p \in \textup{Span}(\cJ(\cA_{u_p}))$.
\end{enumerate}
For each $n\in \mathbb{N}$ and $A\in \{B_{n}^c,B_{n}  \}$ define $V_n=V_{B_n^c}$ and $j_{A}:\cH_{i}\rightarrow \cH_{B_n^c}\oplus \cH_{B_n}$ by $j_{B_n^c}(f)=(f,0)$ and $j_{B_n}f=(0,f)$. Furthermore, we set $Q_{A}=V_n^*j_{A}$ and let $U_n$ be the unitary map from Lemma \ref{Lem:HVZ1} corresponding to $B_n^c$. Fix $f\in \cH$. Then the following equalities holds almost everywhere:
\begin{align}
Q_{B_n} P_{B_n}(f)&=V_n^*(0,P_{B_n}(f))=1_{B_n}P_{B_n}(f)=1_{B_n}f  \label{eq:!!!} \\
Q_{B_n^c} \label{eq:!!!!} P_{B_n^c}(f)&=V_n^*(P_{B_n^c}(f),0)=1_{B_n^c}P_{B_n^c}(f)=1_{B_n^c}f 
\end{align}
 For $f\in \cA_n$ we have $1_{B_n^c}f=f$ and so we obtain the two equalities
\begin{align}\label{eq:cancel1}
\Gamma(Q_{B_n^c})\Gamma(P_{B_n^c})\psi&=\Gamma(1_{B_n^c})\psi=\psi \,\,\, \forall \, \psi\in \text{Span}(\cJ(\cA_n))\\ \label{eq:cancel2}
Q_{B_n}P_{B_n}g^{(i)}_n&=1_{B_n}g^{(i)}_n=g_n^{(i)}
\end{align}
for all $i\in \{1,\dots,\ell\}$. We now define the Weyl sequence as follows:
\begin{align*}
\phi_p=\sqrt{\ell!}U_{u_p}^*(\Gamma(P_{B_{u_p}^c})\psi_p\otimes P_{B_{u_p}}g^{(1)}_{u_p}\otimes_s\dots\otimes_s P_{B_{u_p}}g^{(\ell)}_{u_p})
\end{align*}
We will now prove
	\begin{enumerate}
	\item[\textup{(1)}] $\phi_p\in \cD(H_{\mu}(\xi))$.
	
	\item[\textup{(2)}] $\phi_p$ is orthogonal to $\phi_r$ for $p\neq r$.
	
	\item[\textup{(3)}] $\lVert\phi_p\lVert=1$ for all $p\in \mathbb{N}$.
	
	\item[\textup{(4)}]  $\lVert ( H_\mu(\xi)-\Sigma(\xi-k_0)-\omega^{(\ell)}(k_1,\dots,k_\ell) )\phi_p \lVert $ converges to 0.
	\end{enumerate}
(1): Define for all $p\in \mathbb{N}$ the set 
\begin{align*}
C_p=&\{g_{u_p}^{(1)}\otimes_s\dots\otimes_s g_{u_p}^{(\ell)} \}\\& \cup \bigcup_{q=1}^\infty \{ h_1\otimes_s\cdots\otimes_s h_q\otimes_s g_{u_p}^{(1)}\otimes_s\dots\otimes_s g_{u_p}^{(\ell)}\mid h_i\in \cA_{u_p} \}\subset \cJ(\cC\cS)
\end{align*}
and let $\cK_p:=P_{B_{u_p}^c}\cA_{u_p}\subset\cC\cS_{B_{u_p}^c}$ since $P_{B_{u_p}^c}$ maps $\cC\cS$ into $\cC\cS_{B_{u_p}^c}$. Using equation (\ref{eq:!!!}) we find $Q_{B_{u_p}^c}\cK_p=1_{B_{u_p}}\cA_{u_p}=\cA_{u_p}$ so Lemma \ref{Lem:Splitspace} and equation (\ref{eq:cancel2}) implies 
\begin{equation*}
\psi_p\in \text{Span}(C_p)\subset \text{Span}(\cJ(\cC\cS)) \subset \cD(H_\mu(\xi)).
\end{equation*}
(2): Let $r<p$. Then $\phi_r\in \textup{Span}(C_r)$ and $\phi_p\in \textup{Span}(C_p)$, so we just need to see that every element in $C_p$ and $C_r$ are orthogonal. Let $\psi_1\in C_p$ and $\psi_2\in C_r$. Note every tensor in $C_p$ has a factor $g_{u_p}^{(1)}$ and that this factor is orthogonal to $g_{u_r}^{(i)}$ for all $i$ by construction. Furthermore for any $h\in \cA_{u_r}$ we see that $h$ is supported in $B_{u_r}^c\subset B_{u_p}^c$ and hence $g_{u_p}^{(1)}h=0$, so $g_{u_p}^{(1)}$ is orthogonal to  any element in $\cA_{u_r}$. This implies $\psi_1$ contains a factor orthogonal to all factors in $\psi_2$ and thus $\psi_1$ is orthogonal to $\psi_2$.

(3): $Q_{B_{u_p}^c}$ and $Q_{B_{n_p}}$ are isometric and which implies $\Gamma(Q_{B_{u_p}^c})$ and $\Gamma(Q_{B_{u_p}})$ are isometric. Using equations (\ref{eq:cancel1}) and (\ref{eq:cancel2}) we calculate
\begin{align*}
\lVert \phi_p \lVert&=\sqrt{\ell!}\lVert \Gamma(P_{B_{n_p}^c})\psi_p \lVert \lVert P_{B_{u_p}}g^{(1)}_{u_p}\otimes_s\dots\otimes_s P_{B_{u_p}}g^{(\ell)}_{u_p} \lVert\\&=\sqrt{\ell!}\lVert \Gamma(Q_{B_{u_p}^c})\Gamma(P_{B_{u_p}^c})\psi_p \lVert \lVert \Gamma(Q_{B_{u_p}^c})P_{B_{u_p}}g^{(1)}_{u_p}\otimes_s\dots\otimes_s P_{B_{u_p}}g^{(\ell)}_{u_p} \lVert\\&=\sqrt{\ell!}\lVert \psi_p\lVert \lVert g^{(1)}_{u_p}\otimes_s \dots\otimes_s g^{(\ell)}_{u_p} \lVert=1
\end{align*}
where we used $g^{(i)}_{u_p}$ and $g^{(j)}_{u_p}$ are normalised and orthogonal if $i\neq j$ and
\begin{align*}
	\lVert g^{(1)}_{u_p}\otimes_s \dots\otimes_s g^{(\ell)}_{u_p} \lVert^2&=\frac{1}{\ell!}\sum_{\sigma\in \cS_\ell} \langle g^{(1)}_{u_p}\otimes \dots\otimes g^{(\ell)}_{u_p},g^{(\sigma(1))}_{u_p}\otimes \dots\otimes g^{(\sigma(\ell))}_{u_p}\rangle=\frac{1}{\ell!}.
\end{align*}

(4): Define $g_{u_p}=g^{(1)}_{u_p}\otimes_s \dots\otimes_s g^{(\ell)}_{u_p}$. Using Lemma \ref{Lem:Splitspace} we see that $\lVert  ( H_\mu(\xi)-\Sigma(\xi-k_0)-\omega^{(\ell)}(k_1,\dots,k_\ell) )\phi_p \lVert$ is given by 
\begin{align*}
\sqrt{\ell!} \biggl( &\int_{B_{u_p}^{\ell}}\lVert (H_{B_{u_p}^c}(\xi-x_1-\dots-x_\ell)+\omega^{(\ell)}(x_1,\dots.,x_\ell)\\&-\Sigma(\xi-k_0) -\omega^{(\ell)}(k_1,\dots,k_\ell) )\Gamma(P_{B_{u_p}^c})\psi_p \lVert^2 \lvert g_{u_p}(x)\lvert^2d\lambda_\nu(x)\biggl )^{1/2}:=\sqrt{\ell!} \gamma_p
\end{align*}
Using the triangle inequality, $\lVert \Gamma(P_{B_{n_p}^c})\psi_p \lVert=1$, $ \Gamma(Q_{n_p,B_{u_p}^c})\Gamma(P_{B_{n_p}^c})\psi_p=\psi_p$ and Lemma \ref{Lem:Splitspace} we find $\gamma_p \leq C_1+C_2+C_3$ where
\begin{align*}
C_1&=\left (\int_{B_{u_p}^{\ell}}  \lVert (H_\mu(\xi-x_1-\dots-x_n)-H_\mu(\xi-k_0))\psi_p \lVert^2   \lvert g_{u_p}(x)\lvert^2d\lambda_\nu(x) \right )^{1/2}
\\ C_2 &=\left (\int_{B_{u_p}^{\ell}}  \lvert (\omega^{(\ell)}(x_1,\dots,x_\ell)-\omega^{(\ell)}(k_1,\dots,k_\ell)) \lvert^2   \lvert g_{u_p}(x)\lvert^2d\lambda_\nu(x) \right )^{1/2}\\
C_3&=\lVert (H_\mu(\xi-k_0)-\Sigma(\xi-k_0))\psi_p \lVert \left (\int_{B_{u_p}^{\ell}}    \lvert g_{u_p}(x)\lvert^2d\lambda_\nu(x) \right )^{1/2}
\end{align*}
Let $f:(\RR^\nu)^n\rightarrow \RR_+$ be non negative and symmetric. Using that the $g_{u_p}^{(i)}$ have disjoint support one finds
\begin{align*}
\lvert g_{u_p}(x_1,\dots,x_\ell)\lvert^2=\frac{1}{\ell!^2}\sum_{\sigma, \pi \in \cS_n} \prod_{i=1}^\nu \overline{ g_{u_p}^{(\pi(i))}(x_i)}g_{u_p}^{(\sigma(i))}(x_i)=\frac{1}{\ell!^2}\sum_{\sigma \in \cS_n} \prod_{i=1}^\nu \lvert g_{u_p}^{(\sigma(i))}(x_i)\lvert^2
\end{align*}
Thus using permutation invariance of $f$ we find
\begin{align*}
\int_{B_{u_p}^{\ell}}  f(x) \lvert g_{u_p}(x)\lvert^2d\lambda_\nu(x)=\frac{1}{\ell!} \int_{A_{u_p}}  f(x)  \prod_{i=1}^\nu \lvert g_{u_p}^{(i)}(x_i)\lvert^2 d\lambda_\nu(x)
\end{align*}
Thus $\sqrt{\ell!}C_3=\lVert (H_\mu(\xi-k_0)-\Sigma(\xi-k_0))\psi_p \lVert \leq p^{-1}$. Furthermore
\begin{align*}
&\sqrt{\ell!}C_1\leq \sup_{(x_1,\dots,x_n)\in A_{u_p}} \lVert (H_\mu(\xi-x_1-\dots-x_\ell)-H_\mu(\xi-k_0))\psi_p \lVert\leq p^{-1}\\
&\sqrt{\ell!}C_2\leq \sup_{(x_1,\dots,x_n)\in A_{u_p}} \lvert \omega^{(\ell)}(x_1,\dots,x_\ell)-\omega^{(\ell)}(k_1,\dots,k_\ell) \lvert
\end{align*}
By continuity of $\omega$ we now see $\sqrt{\ell!}\gamma_p$ goes to 0 for $p$ tending to $\infty$.\cqfd
\end{proof}
\begin{lemma}\label{Lem:HVZ}
Let $k_1,\dots,k_\ell\in \RR^\nu$. Then $\Sigma(\xi-k_1-\dots-k_\ell)+\omega^{(\ell)}(k_1,\dots,k_\ell)\in \sigma_{ess}(H_\mu(\xi))$.
\end{lemma}
\begin{proof}
Assume first $k_1,\dots,k_\ell\in \RR^\nu$ are different elements and define $A_n=B_{1/n}(k_1)\cup\dots\cup B_{1/n}(k_\ell)$. Let $v_n=1_{A_n^c}v$ and note that $v_n\in \cD(\omega^{-1/2})$ and
\begin{align*}
\lim_{n\rightarrow \infty }\lVert (v_n-v)(\omega^{-1/2}+1)\lVert=0
\end{align*}
by dominated convergence. Define
\begin{align*}
H^{(n)}(\xi)&=K(\xi-d\Gamma(m))+d\Gamma(\omega)+\mu \varphi(v_n)\geq -\mu^2\lVert \omega^{-1/2}v_n \lVert^2\geq  -\mu^2\lVert \omega^{-1/2}v \lVert^2\\ \Sigma_n(\xi)&=\inf(\sigma(H^{(n)}(\xi)))
\end{align*}
Using Lemma \ref{Lem:FundamentalIneq} and standard resolvent formulas we find
\begin{align*}
\lVert (H_\mu(\xi)+i)^{-1}&-(H^{(n)}(\xi)+i)^{-1} \lVert\leq \lvert \mu\lvert  \lVert \varphi(v-v_n)(H_\mu(\xi)+i)^{-1})^{-1}\lVert \\& \leq \lvert \mu\lvert  \lVert (v_n-v)(\omega^{-1/2}+1)\lVert \lVert (d\Gamma(\omega)+1)^{1/2}(H_\mu(\xi)+i)^{-1}\lVert
\end{align*}
so $H^{(n)}(\xi)$ converges to $H_\mu(\xi)$ in norm resolvent sense for all $\xi\in \RR^\nu$. The uniform lower bound of $\Sigma_n(\xi)$ and norm resolvent convergence now implies $\Sigma_n(\xi)$ converges to $\Sigma(\xi)$ for all $\xi$ (see \cite[Lemma 5.5]{SB2}).

By Lemma \ref{Lem:HVZ1} we have $\Sigma_n(\xi-k_1-\dots-k_\ell)+\omega^{(\ell)}(k_1,\dots,k_\ell)\in \sigma_{ess}(H^{(n)}(\xi))$. Now $\Sigma_n(\xi-k_1-\dots-k_\ell)+\omega^{(\ell)}(k_1,\dots,k_\ell)$ converges to $\Sigma(\xi-k_1-\dots-k_\ell)+\omega_n(k_1,\dots,k_\ell)$ and $H^{(n)}(\xi)$ converges to $H_\mu(\xi)$ in norm resolvent sense so we are done in the case where $k_1,\dots,k_\ell$ are different. The conclusion now follows since $\Sigma$ and $\omega^{(\ell)}$ are continuous, $\{ (k_1,\dots,k_\ell)\mid k_i\neq k_j\,\,\, \forall i,j  \}$ is dense and $\sigma_{ess}(H_\mu(\xi))$ is closed. \cqfd
\end{proof}

\section{Proof of pull though formula}
This appendix is devoted to proving the pull through formula. In case $K(k)=\lvert k\lvert^2$ one could compute everything directly using tools as in \cite{HerbstHasler}. However the other possible choices of $K$ require a more sophisticated approach ao we use the formalised developed in \cite{Thomas1}. We give a brief introduction here but the reader should consult \cite{Thomas1} for the proofs. We start by defining
\begin{equation*}
\cF_{+}(\cH)=  \bigtimes_{n=0}^{\infty}\cH^{\otimes_s n}
\end{equation*}
with coordinate projections $P_n$ and $\cH=L^2(\RR^\nu,\cB(\RR^\nu),\lambda_\nu)$. For $(\psi^{(n)}),(\phi^{(n)})\in \cF_+(\cH)$ we define
\begin{equation*}
d((\psi^{(n)}),(\phi^{(n)}))=\sum_{n=0}^{\infty} \frac{1}{2^n}\frac{\lVert \psi^{(n)}-\phi^{(n)} \lVert  }{1+\lVert \psi^{(n)}-\phi^{(n)} \lVert }
\end{equation*} 
where $\lVert \cdot \lVert$ is the Fock space norm. This makes sense since $P_n(\cF_{+}(\cH))\subset \cF(\cH)$. We now have
\begin{lemma}\label{Lem:BasicTopologyExtSpace}
The map $d$ defines a metric on $\cF_{+}(\cH)$ and turns this space into a complete separable metric space and a topological vector space. The topology and Borel $\sigma$-algebra is generated by the projections $P_n$. 
\end{lemma}
Fix now $v\in \cH$. We now define the following maps on $\cF_{+}(\cH)$
\begin{align*}
a_+(v) ( \psi^{(n)} ) &= ( a_{n}(v)\psi^{(n+1)} )\\
a^{\dagger}_+(v) ( \psi^{(n)} ) &= (0, a^\dagger_0(v)\psi^{(0)},a_1^\dagger(v)\psi^{(1)},\dots )\\
\varphi_+(v)  &= a_+(v)+a_+^\dagger(v)
\end{align*}
Where $a_n(v)$ is annihilation from $\cH^{\otimes_s (n+1)}$ to $\cH^{\otimes_s n}$ and $a^\dagger_n(f)$ is creation from $\cH^{\otimes_s n}$ to $\cH^{\otimes (n+1)}$. 
\begin{lemma}\label{Lem:AnihilCreaFieldOnF_+}
	The maps $a_+(v)$,$a^{\dagger}_+(v)$ and $\varphi_+(v)$ are all continuous. For $B\in \{ a,a^{\dagger},\varphi \}$ we have
	\begin{equation}\label{eq:extphi}
	B_+(v)\psi=B(v)\psi \,\,\, \text{if} \,\,\, \psi\in \cD(B(v)).
	\end{equation}
\end{lemma}
Let $\omega:\RR^\nu\rightarrow \RR^p$ be measurable and write it in terms of it coordinate functions as $\omega=(\omega_1,\dots,\omega_p)$. We then define
\begin{align*}
d\Gamma(\omega)&=(d\Gamma(\omega_1),\dots,d\Gamma(\omega_p))\\
d\Gamma^{(n)}(\omega)&=(d\Gamma^{(n)}(\omega_1),\dots,d\Gamma^{(n)}(\omega_p))
\end{align*}
Let $f:\RR^p\rightarrow \CC$ be a map and define
\begin{align*}
f(d\Gamma_+(\omega))&=\bigtimes_{n=0}^{\infty}f(d\Gamma^{(n)}(\omega)) \,\,\,\, \cD(f(d\Gamma_+(\omega)))=\bigtimes_{n=0}^{\infty}\cD(f(d\Gamma^{(n)}(\omega))) 
\end{align*}
We now have
\begin{lemma}\label{Lem:2ndQuantisedF_+}
	The following identity holds for all $\psi\in \cD(f(d\Gamma(\omega)))$
	\begin{equation*}
	f(d\Gamma_+(\omega))\psi=f(d\Gamma(\omega))\psi, \,\,\,\,\, \psi\in \cD(f(d\Gamma(\omega)))
	\end{equation*}
\end{lemma}
\noindent We will need the following spaces: For each $a\in \RR$ we define
\begin{equation*}
\lVert \cdot \lVert_{a,+}=\lim_{n\rightarrow \infty}\left (\sum_{k=0}^{n} (k+1)^{2a} \lVert P_k(\cdot) \lVert^2   \right )^{\frac{1}{2}}.
\end{equation*}
which is measurable from $\cF_+(\cH)$ into $[0,\infty]$. Let
\begin{align*}
\cF_{a,+}(\cH)=\{ \psi\in \cF_{+}(\cH)\mid \lVert \psi \lVert_{a,+}<\infty  \}.
\end{align*}

\noindent We will now consider a class of linear functionals on $\cF_+(\cH)$. For each $n \in \mathbb{N}$ we let $Q_n:\cF_{+}(\cH)\rightarrow \cN$ denote the linear projection which preserves the first $n$ entries of $(\psi^{(n)})$ and projects the rest of them to 0. For $\psi\in \cN$ there is $K\in \mathbb{N}$ such that for $n\geq K$ we have $Q_n\psi=\psi$. For $\phi\in \cF_+(\cH)$ we may thus define the pairing
\begin{equation}\label{eq:extinner}
\langle \psi,\phi \rangle_+:=\langle \psi,Q_n\phi \rangle=\sum_{i=0}^{K} \langle \psi^{(i)}, \phi^{(i)} \rangle,
\end{equation}
where $n\geq K$.
\begin{lemma}\label{Lem:SeperatingForm}
	The map $Q_n$ above is linear and continuous into $\cF(\cH)$. The paring $\langle \cdot, \cdot \rangle_+$ is sesquilinear, and continuous in the second entry. If $\phi\in \cF_{a,+}(\cH)$ then $\psi\mapsto \langle \psi,\phi \rangle_+$ is continous with respect to $\lVert \cdot \lVert_{-a,+}$. Furthermore, the collection of maps of the form $\langle \psi, \cdot \rangle_+$ will separate points of $\cF_+(\cH)$.
\end{lemma}
\begin{corollary}\label{Cor:seperatingDense}
	Let $\phi\in \cF_{a,+}(\cH)$ for some $a\leq 0$, $\cD\subset \cN$ be dense in $\cF(\cH)$ and assume $\langle \psi, \phi \rangle_+=0$ for all $\psi\in \cD$. Then $\phi=0$.
\end{corollary}
We also have the following formal adjoint relations
\begin{lemma}\label{Lem:FormalAdjoints}
	Let $\psi\in \cN$, $\phi\in \cF_+(\cH)$, $v\in \cH$ and $U$ be unitary on $\cH$. Then we have
	\begin{align*}
	\langle a^{\dagger}(v)\psi,\phi \rangle_+&=\langle \psi,a_+(v) \phi \rangle_+,\,\,\,\,\,\,\,\,\,\,\,\,\,\,\, \langle a(v)\psi,\phi \rangle_+=\langle \psi,a_+^{\dagger}(v) \phi \rangle_+,\\ \langle \varphi(v)\psi,\phi \rangle_+&=\langle \psi,\varphi_+(v) \phi \rangle_+,\,\,\,\,\,\,\,\,\,\,\,\,\,\,\, \langle \Gamma(U)\psi,\phi \rangle_+=\langle \psi,\Gamma_+(U^*)\phi \rangle_+.
	\end{align*}
	Let $\omega:\RR^\nu\rightarrow \RR^p$ be a measurable map, $f:\RR^p\rightarrow \CC$, $\psi \in \cN\cap \cD( f(d\Gamma(\omega))  )$ and $\phi\in \cD(\overline{f}(d\Gamma_+(\omega))  )$. Then
	\begin{equation*}
	\langle f(d\Gamma(\omega))\psi,\phi \rangle_+=\langle \psi,\overline{f}(d\Gamma_+(\omega))\phi \rangle_+.
	\end{equation*}
\end{lemma}
\noindent We now consider functions with values in $\cF_+(\cH)$. Define the quotient
\begin{equation*}
\cM=\{f:\RR^\nu \rightarrow  \cF_+(\cH)\mid f \,\, \text{is}\,\, \cB(\RR^\nu)-\cB(\cF_+(\cH))\,\, \text{mesurable}  \}/\sim,
\end{equation*}
where we define $f\sim g\iff f=g$ almost everywhere. We are interested in the subspace
\begin{equation*}
\cM_{\cI}=\{ f\in \cM \mid x\mapsto P_nf(x)\in L^2(\RR^\nu,\cB(\RR^\nu),\lambda_\nu,\cH^{\otimes_s n}) \,\, \forall \, n\in \NN_0 \}.
\end{equation*}
We write $f\in \cM_{\cI}$ as $(f^{(n)})$ where $f^{(n)}=x\mapsto P_nf(x)$. For $f,g\in \cM_{\cI}$ we define
\begin{equation*}
d(f,g)=\sum_{n=0}^{\infty} \frac{1}{2^n}\frac{\lVert f^{(n)}-g^{(n)} \lVert_{L^2(\RR^\nu,\cB(\RR^\nu),\lambda_\nu,\cH^{\otimes_s n})}  }{1+\lVert f^{(n)}-g^{(n)} \lVert_{L^2(\RR^\nu,\cB(\RR^\nu),\lambda_\nu,\cH^{\otimes_s n})} }.
\end{equation*} 
\begin{lemma}\label{Lem:TheFundamentalMeasureSpace}
	$d$ is a complete metric on $\cM_{\cI}$ such that $\cM_{\cI}$ becomes separable topological vector space.
\end{lemma}
\noindent Let $v\in \cH$, $\omega=(\omega_1,\dots,\omega_p)$ a tuple of selfadjoint multiplication operators on $\cH$, $m:\RR^\nu \rightarrow \RR^p$ measurable and $g:\RR^p\rightarrow \RR$ a measurable map. Then we wish to define operators on $\cM_{\cI}$ by
\begin{align*}
(a^{\dagger}_{\oplus}(v)f)(k)&=a^{\dagger}_+(v)f(k)\\
(a_{\oplus}(v)f)(k)&=a_+(v)f(k)\\
(\varphi_{\oplus}(v)f)(k)&=\varphi_+(v)f(k)\\
(g(d\Gamma_{\oplus}(\omega)+m)f)(k)& =g(d\Gamma_+(\omega)+m(k))f(k).
\end{align*} 
We have the following lemma.
\begin{lemma}\label{Lem:LiftToIntegralOperators}
	$a^{\dagger}_{\oplus}(v),a_{\oplus}(v)$ and $\varphi_{\oplus}(v)$ are well defined and continuous. Let $f\in \cM_{\cI}$. If $f(k)\in \cD(g(d\Gamma_+(\omega)+m(k)))$ for all $k$ then 
	\begin{equation*}
	k\mapsto P_n(g(d\Gamma_+(\omega)+m(k))f(k))
	\end{equation*}
	Thus as domain of $g(d\Gamma_{\oplus}(\omega)+m)$ we may choose
	\begin{align*}
	\bigcap_{\ell=0}^\infty \biggl \{ f\in \cM_{\cI}\biggl \lvert  &f(k)\in \cD(g(d\Gamma_+(\omega)+m(k)))\,\, \text{for a.e.} \,\, k\in \RR^\nu,\\& \int_{\cM^\ell} \lVert P_n g(d\Gamma_+(\omega)+m(k))f(k) \lVert^2 d\mu^{\otimes \ell}(k)<\infty  \biggl \}.
	\end{align*}
\end{lemma}
We will now introduce the pointwise annihilation operators. For $\psi=(\psi^{(n)}) \in \cF_+(\cH)$ we define $A\psi\in \cM_{\cI}$ by
\begin{equation*}
P_n(A\psi(k))=\sqrt{(n+1)}\psi^{(n+1)}(k,\cdot,\dots,\cdot).
\end{equation*}
\begin{lemma}\label{Lem:NumberEstimates}
	$A$ is a continuous linear map from $\cF_+(\cH)$ to $\cM_{\cI}$ and if $\psi\in \cF(\cH)$ then $A\psi$ is almost everywhere $ \cF_{-1/2,+}(\cH)$ valued. 
\end{lemma}
\noindent Fix $v\in \cH$ and define a map $z(v):\cF_{+}(\cH)\rightarrow \cM_{\cI}$ by 
\begin{equation*}
(z(v)\psi)(k)=v(k)\psi.
\end{equation*}
\begin{lemma}\label{Lem:ContOfMultOperator}
	The map $z(v)$ is linear and continuous from $\cF_{+}(\cH)$ into $\cM_{\cI}$.
\end{lemma}

\noindent We have the following relations (more relations can be found in \cite{Thomas1} but we will only cite those used here)

\begin{lemma}\label{Lem:CommutatingPointwiseAnihilation}
	Let $\omega:\RR^\nu \rightarrow \RR^p$ be measurable, $v\in \cH$ and let $f:\RR^p\rightarrow \RR$ be measurable. Then
	\begin{equation*}
	\varphi_\oplus(v)A=A\varphi_+(v)-z(v)
	\end{equation*}
	If $\psi\in \cD(f(d\Gamma(\omega)))$ then $A\psi\in \cD(f(d\Gamma_{\oplus}(\omega)+\omega) )$ and
	\begin{equation*}
	f(d\Gamma_{\oplus}(\omega)+\omega) A\psi=A f(d\Gamma_{+}(\omega))\psi.
	\end{equation*}
\end{lemma}
\noindent We can now prove the pull-trough formula.
\begin{lemma}\label{Lem:pullthr}
	Assume $\omega,v,K$ satisfy Hypotheses 1 and 2 and let $\mu\in \RR,\xi \in \RR^\nu,\nu\geq 2$. Assume $\psi\in \cD(H_\mu(\xi))$ and $A(H_\mu(\xi)-\Sigma(\xi))\psi$ is fockspace valued. Then we have
	\begin{align*}\label{PullthroughFormula}
	(A\psi)(k)=&(H_\mu(\xi-k)+\omega(k)-\Sigma(\xi) )^{-1}A((H_\mu(\xi)-\Sigma(\xi))\psi)(k)\\&-\mu v(k)(H_\mu(\xi-k)+\omega(k)-\Sigma(\xi) )^{-1}\psi
	\end{align*}
	almost everywhere.
\end{lemma}
\begin{proof}
	First we note $(H_\mu(\xi-k)+\omega(k)-\Sigma(\xi) )^{-1}$ exists as a bounded operator away from the zero set $\RR \xi$ by Lemma \ref{Lem:Energyineq}. Define the lifted operators on $\cF_+(\cH)$ and $\cM_{\cI}$ respectively
	\begin{align*}
	H_{+}(\xi)&=K(\xi-d\Gamma_+(m))+d\Gamma_+(\omega)+\mu\varphi_+(v)\\
	H_{\oplus}(\xi)&=K(\xi-m-d\Gamma_{\oplus}(m))+d\Gamma_{\oplus}(\omega)+\omega +\mu\varphi_{\oplus}
	\end{align*}
	where $m:\RR^\nu\rightarrow \RR^\nu$ is given by $m(k)=k$. The domains are
	\begin{align*}
	\cD(H_{+}(\xi))&=\cD(d\Gamma_+(\omega))\cap \cD(K(\xi-d\Gamma_+(m)))\\
	\cD(H_{\oplus}(\xi))&=\cD(d\Gamma_{\oplus}(\omega)+\omega)\cap \cD(K(\xi-m-d\Gamma_{\oplus}(m)))
	\end{align*}
	By Lemma \ref{Lem:CommutatingPointwiseAnihilation} we have $A\psi \in \cD(H_{\oplus}(\xi))$ since $\psi\in \cD(H_\mu(\xi))\subset \cD(H_{+}(\xi))$. Using Lemmas \ref{Lem:AnihilCreaFieldOnF_+}, \ref{Lem:2ndQuantisedF_+} and  \ref{Lem:CommutatingPointwiseAnihilation} we also obtain
	\begin{align*}
	h:&=(H_{\oplus}(\xi)-\Sigma(\xi))A\psi\\&=-\mu z(v)\psi+A(H_{+}(\xi)-\Sigma(\xi))\psi\\&=-\mu z(v)\psi+A(H_\mu(\xi)-\Sigma(\xi))\psi
	\end{align*}
	which is Fock space valued. Let $M$ be a zeroset such that:
	\begin{enumerate}
		\item $A\psi$ is $\cF_{-1/2,+}(\cH)$ valued on $M^c$ (see Lemma \ref{Lem:NumberEstimates}).
		\item $h(k)=(H_+(\xi-k) +\omega(k) )(A\psi(k))$ and $h(k)\in \cF(\cH)$ for $k\in M^c$. 
		\item $(H_\mu(\xi-k)+\omega(k)-\Sigma(\xi) )^{-1}$ exists on $M^c$.
	\end{enumerate}
	Fix $k\in M^c$. For any vector $\phi$ such that both $(H_\mu(\xi-k)+\omega(k)-\Sigma(\xi) )^{-1}\phi$ and $\phi$ is in $\cN$ (this set is dense by Proposition \ref{Lem:FundamentalTechnicalStuff}) we find using Lemma \ref{Lem:FormalAdjoints} that
	\begin{align*}
	\langle &\phi, A\psi (k) \rangle_+\\&=\langle (H_\mu(\xi-k)+\omega(k)-\Sigma(\xi)) (H_\mu(\xi-k)+\omega(k)-\Sigma(\xi))^{-1}\phi, A\psi (k) \rangle_+\\&=\langle (H_\mu(\xi-k)+\omega(k)-\Sigma(\xi))^{-1}\phi, h(k) \rangle\\&=\langle \phi, (H_\mu(\xi-k)+\omega(k)-\Sigma(\xi))^{-1}h(k) \rangle_+.
	\end{align*}
	Corollary \ref{Cor:seperatingDense} finishes the proof.
\end{proof}

\end{document}